\documentclass[12pt,a4paper,reqno]{amsproc}
\usepackage[margin=0.8in]{geometry}
\pdfoutput=1
\makeatletter
\def\@tocline#1#2#3#4#5#6#7{\relax
  \ifnum #1>\c@tocdepth 
  \else
    \par \addpenalty\@secpenalty\addvspace{#2}%
    \begingroup \hyphenpenalty\@M
    \@ifempty{#4}{%
      \@tempdima\csname r@tocindent\number#1\endcsname\relax
    }{%
      \@tempdima#4\relax
    }%
    \parindent\z@ \leftskip#3\relax \advance\leftskip\@tempdima\relax
    \rightskip\@pnumwidth plus4em \parfillskip-\@pnumwidth
    #5\leavevmode\hskip-\@tempdima
      \ifcase #1
       \or\or \hskip 1em \or \hskip 2em \else \hskip 3em \fi%
      #6\nobreak\relax
    \dotfill\hbox to\@pnumwidth{\@tocpagenum{#7}}\par
    \nobreak
    \endgroup
  \fi}
\makeatother
\usepackage{xcolor}
\usepackage{soul, mathtools}
\usepackage{graphicx} 
\usepackage{cite}
\usepackage{rotating}
\usepackage{amsmath}
\usepackage{amssymb}
\usepackage{amsfonts}
\usepackage{amsthm}
\usepackage{mathrsfs, mathtools}
\usepackage{bm}
\usepackage[utf8]{inputenc}
\usepackage{tikz-cd}
\usepackage{soul}
\usepackage{bbm}
\usepackage[bookmarks, colorlinks, breaklinks]{hyperref} 
\hypersetup{
	linkcolor=blue,
	citecolor=blue,
	filecolor=black,
	urlcolor=blue}
\makeatletter
\renewcommand*{\eqref}[1]{%
  \hyperref[{#1}]{\textup{\tagform@{\ref*{#1}}}}%
}
\makeatother
\makeatletter
\usepackage[capitalise,nameinlink]{cleveref}
\renewcommand*{\cref}[1]{%
  \hyperref[{#1}]{\textup{\tagform@{\ref*{#1}}}}%
}
\makeatother

\numberwithin{equation}{section}

\renewcommand{\paragraph}[1]{\textbf{#1}:}
\theoremstyle{definition}
\newtheorem{thm}{Theorem}[section]

\newtheorem{props}[thm]{Proposition}
\newtheorem{lemma}[thm]{Lemma}
\newtheorem{rmk}[thm]{Remark}
\newtheorem{cor}[thm]{Corollary}
\newtheorem{theorem}[thm]{Theorem}

\newcommand{\taub}{\bar{\tau}}
\newcommand{\etabar}{\overline{\eta}}
\newcommand{\taubar}{\overline{\tau}}
\newcommand{\qbar}{\overline{q}}
\newcommand{\cbar}{\overline{c}}

\newcommand{\fstop}{\, .}

\newcommand{\zbar}{\overline{z}}
\newcommand{\mbar}{\overline{m}}
\newcommand{\nbar}{\overline{n}}
\newcommand{\jbar}{\overline{J}}
\newcommand{\dtil}{d^*}

\newcommand{\ket}[1]{|#1 \rangle}

\newcommand{\disq}{\mathscr{D}_\Lambda} 
\newcommand{\gentheta}[2]{\vartheta\!\left[ \genfrac{}{}{0pt}{}{#1}{#2}\right]\!}
\newcommand{\mlatt}{\mathcal{P}_\Lambda} 
\newcommand{\nviel}{\mathcal{E}} 
\newcommand{\repsp}{\mathfrak{h}_Q} %
\newcommand{\modsp}{\mathcal{M}_\Lambda} 
\newcommand{\genTD}{ \mathrm{O}_{Q}(p,q;\mathbb{Z})} 
\newcommand{\genTDa}{ \mathrm{O}_{Q,\alpha}(p,q;\mathbb{Z})} 
\newcommand{\mlattn}{\mathcal{P}_{\rm Narain}} 
\newcommand{\modspn}{\mathcal{M}_{\rm Narain}} 
\newcommand{\narTD}{ \mathrm{O}(D,D;\mathbb{Z})} 

\def\R{\mathbb{R}}
\def\C{\mathbb{\C}}

\newcommand{\SL}{\mathrm{SL}}
\newcommand{\PSL}{\mathrm{PSL}}
\newcommand{\GL}{\mathrm{GL}}
\newcommand{\rmO}{O}

\newcommand{\SU}{\mathrm{SU}}

\newcommand{\tl}{\tilde{l}}
\newcommand{\tz}{\tilde{z}}
\newcommand{\myS}{\zeta}
\newcommand{\jac}[2]{\left(\dfrac{#1}{#2} \right)}

\definecolor{DESYO}{RGB}{241,143,31}
\definecolor{DESYB}{RGB}{0,159,223}
\definecolor{DESYC}{RGB}{153,0,18}

\def\ie{\begin{equation}\begin{aligned}}
\def\fe{\end{aligned}\end{equation}}

\begin{document}

\title[Generalized Narain Theories Decoded]{{\Huge Generalized Narain Theories $\mathfrak{Decoded}$:}\\ \medskip {$\mathfrak{D}$iscussions on $\mathfrak{E}$isenstein series, $\mathfrak{C}$haracteristics, $\mathfrak{O}$rbifolds, $\mathfrak{D}$iscriminants \& $\mathfrak{E}$nsembles} in any $\mathfrak{D}$imension}
\author[Ashwinkumar, Kidambi, Leedom, Yamazaki]{\large Meer Ashwinkumar$^{N,a}$, Abhiram Kidambi$^{r,a}$, \\ Jacob M. Leedom$^{i}$, Masahito Yamazaki$^{n,a}$}
\begin{titlepage}
\vspace*{-1cm}
\hfill DESY 23-170\\
\vspace{1cm}
\maketitle

\begin{center}
\noindent \textit{$^N$ Albert Einstein Center for Fundamental Physics, Institute for Theoretical Physics \\ University of Bern, Sidlerstrasse 5, CH-3012 Bern, Switzerland
 }\\[0.2cm]
 \textit{$^a$ Kavli IPMU, Uni. Tokyo, Kashiwanoha 5-1-5, 277-8583 Kashiwa, Chiba, Japan
 }\\[0.2cm]
 \textit{$^r$
Erwin Schr\"odinger Institute, Boltzmanngasse 1A, 1090 Vienna, Austria
 }\\[0.2cm]
 \textit{$^i$ Deutsches Elektronen-Synchrotron DESY, Notkestr. 85, 22607 Hamburg, Germany
 }
\\[0.2cm]
 \textit{$^n$ Trans-Scale Quantum Science Institute, The University of Tokyo, Tokyo 113-0033, Japan 
 }
 \\[0.3cm]
 { MA: \href{mailto:meer.ashwinkumar@unibe.ch}{meer.ashwinkumar@unibe.ch} 
   AK: \href{mailto:kidambi@duck.com}{kidambi@duck.com} \\  
   JML: \href{mailto:jacob.michael.leedom@desy.de}{jacob.michael.leedom@desy.de} 
   MY: \href{mailto:masahito.yamazaki@ipmu.jp}{masahito.yamazaki@ipmu.jp}\\
[1cm]}
 \end{center}
 \begin{center}
     \textbf{Abstract}\\[1em]
 \end{center}
We study a class of newly-introduced CFTs associated with even quadratic forms of general signature,
which we call \textit{generalized Narain theories}. We first summarize the properties of these theories. We then consider orbifolds of these theories, thereby obtaining a large class of non-supersymmetric CFTs with exactly marginal deformations. We then discuss ensemble averages of such theories over their moduli space,
and obtain a modular form associated with the quadratic form and an element of the discriminant group. 
The modular form can be written as a Poincar\'e series, which contains novel invariants of lens spaces and suggests the interpretation of the holographic bulk as a theory of anyons. 
\thispagestyle{empty}
\end{titlepage}


\tableofcontents

\section{Introduction}
\label{sec:intro}

In physics one often makes progress by studying simple models
which capture the essence of complicated physical phenomena.
This comment may well be true in the theories of holography:
instead of complicated full-fledged string theory setups for holography, 
one can hope to make progress by first studying simpler solvable setups in detail and then try to embed the resulting setups back into string theory.

Ensemble averages of Narain Conformal Field Theories (CFTs)
as recently discussed in \cite{Afkhami-Jeddi:2020ezh,Maloney:2020nni} (see also \cite{Perez:2020klz,Dymarsky:2020bps,Dymarsky:2020pzc,Meruliya:2021utr,Datta:2021ftn,Benjamin:2021wzr,Meruliya:2021lul,Eberhardt:2021jvj,Kames-King:2023fpa})
can be regarded as perfect examples of such simplified models,
especially when we wish to explore ensemble averages in holography \cite{Saad:2019lba,Stanford:2019vob}. 
The ensemble averages of 2d CFTs can be computed exactly 
over the CFT moduli space, and the result in the bulk generates an exotic theory of quantum gravity, where
we have a sum over geometries and the contribution from each geometry is captured by an Abelian Chern-Simons theory.
It was moreover suggested that these holographic dualities can be embedded into string theory \cite{Ashwinkumar:2023jtz}\footnote{See e.g.\ \cite{Eberhardt:2021jvj,Heckman:2021vzx,Collier:2022emf,Baume:2023kkf} for a sample of papers which discuss the embedding of ensemble averages into string theory, albeit in different setups.}, based on an earlier claim of \cite{Gukov:2004id}.

In our previous paper \cite{Ashwinkumar:2021kav} we analyzed ensemble averages of generalized Narain CFTs associated with a general even integer lattice. Since the lattice is in general not unimodular, the torus partition function is not modular invariant, and this is probably one of the reasons why such theories have rarely been discussed in string theory literature.\footnote{In the literature, modular invariance is sometimes included in the definition of the two-dimensional CFT, and most of our theories are not CFTs under this definition. We find it useful, however, to be more flexible in the definition of the CFTs,
at least for the purposes of this paper. Note that even under the strict definition requiring modular invariance, our theories can be used as building blocks for modular-invariant theories, as discussed in~\Cref{sec:narain}.}
The generalized CFTs, however, are otherwise well-defined theories, and we can average over their
CFT moduli spaces, to obtain an exotic theory of gravity in the bulk described by a three-dimensional Abelian Chern-Simons theory. Here the non-modular-invariance of the theory 
is accounted for by the existence of the Chern-Simons term in the bulk. Moreover, the non-modular-invariance implies that we have a non-trivial set of anyons in the holographic bulk, and this leads to global symmetries in the bulk, some of which are emergent only 
after the ensemble averaging \cite{Ashwinkumar:2023jtz}.

In this paper we discuss ensemble averages of {\it orbifolds} of the Narain CFT \cite{Narain:1985jj,Narain:1986am} associated with a general quadratic form, generalizing the previous analysis \cite{Ashwinkumar:2021kav} of the un-orbifolded case
and \cite{Benjamin:2021wzr,Dong:2021wot} for even self-dual quadratic forms.\footnote{See also \cite{Kames-King:2023fpa} for another paper on ensemble averages for orbifolded Narain theories for special choices of unimodular quadratic forms.}
We will also discuss generalizations of the aforementioned ensemble averages where chemical potentials for flavor symmetries are turned on in the Narain CFT partition function.  

There are several motivations for studying orbifolds of generalized Narain CFTs.
Firstly, by studying orbifolds we can construct a new class of irrational (and rational) CFTs.
Secondly, it serves to understand the less-studied CFTs associated to general (indefinite) even quadratic forms.
Thirdly, it is rare to find non-supersymmetric CFTs with exactly marginal deformations,
and orbifolded Narain CFTs will provide further concrete examples to test the ideas of ensemble averages in holography.
Fourthly, orbifolding gauges part of the T-duality symmetry, and this can be regarded as 
gauging of the ``ensemble symmetries'' in holography, as articulated in \cite{Ashwinkumar:2023jtz} (see also \cite{Antinucci:2023uzq}).
Note that global symmetries are present only in exotic holographies involving non-Einsteinian gravity, such as the ones discussed for generalized Narain theories. Lastly, our analysis generates automorphic forms which are generalizations of the Siegel-Eisenstein series, and could be of independent mathematical interest.

The rest of this paper is organized as follows.
In \Cref{sec:narain} we first summarize generalized Narain theories of \cite{Ashwinkumar:2021kav}.
In \Cref{sec:orbifold} we discuss orbifolds of the Narain theories.
In \Cref{sec:ensemble} we then discuss ensemble averages of the torus partition functions
over the CFT moduli space.  We will also extend the discussion of ensemble averages by including chemical potentials for the flavor symmetries, and will encounter new Siegel-Jacobi forms associated with quadratic forms.
In \Cref{sec:bulk} we briefly comment on the holographic-dual theories.
We conclude with summaries and discussions in \Cref{sec:discussion}.
The appendices contain technical materials
needed for the understanding of the main text.

\section{Generalized Narain Theories}\label{sec:narain}

In this section, we describe the generalized Narain theories,
which contains the original Narain CFTs (toroidal CFTs) \cite{Narain:1985jj,Narain:1986am} as special examples.
While we mostly follow our previous paper \cite{Ashwinkumar:2021kav}, we expand the discussion to clarify the construction. 
Our conventions and notations are summarized in \Cref{appendix:conventions}. 

\subsection{Review of Standard Narain CFTs}
\label{sec:stdnarain}

We begin with a recasting of a familiar story---Narain CFTs. Narain CFTs can be constructed from toroidal compactifications of superstring theories, in particular the heterotic string\footnote{For a nice discussion whose conventions we follow, see~\cite{Nilles:2021glx}.}~\cite{Narain:1985jj,Narain:1986am}. These theories are defined via Narain lattices $\mlattn$ that are even and self-dual with respect to the inner product. 
For now, we focus on Narain CFTs with equal left- and right-moving central charges $c_L=c_R=D$ so that the signature of $\mlattn$ is $(D+,D-)$. The momenta $(p_L,p_R)$ depend on the moduli of the compactification: the background metric $G$, and the two-form $B$.
These moduli parameterize the marginal deformations of the CFT and take values in the typical Narain moduli space $\modspn$ defined below in~\cref{eq:narainmod}. At a point $m\in\modspn$, the components of a lattice vector $(p_L,p_R)\in\mlattn$ can be written\footnote{For $c_L\neq c_R$, there are also Wilson line moduli that will further modify these expressions.} 
\begin{align}
\begin{split}
    p_{L,i} &=  n_i+ \frac{1}{2}(G_{ij} - B_{ij})w^j  \;, \\ 
    p_{R,i} &=  n_i - \frac{1}{2}(G_{ij} + B_{ij})w^j  \;,
\end{split}
\end{align}
with $n_i, w^j\in \mathbb{Z}$ ($i,j = 1,\cdots D$).
The operator content of the Narain CFT consists of currents $J^M = \partial X^M$ and $\bar{J}^M = \bar{\partial}X^M$ and vertex operators $\mathcal{V} =\; :\exp\bigg(i p_L\cdot X_L + i p_R\cdot X_R\bigg):$, where the $X^M = X^M_L + X^M_R$ ($M=1,\dots ,D$) are the compact bosons of the CFT. 

The above can be reformulated in a more suggestive manner. The Narain lattice consists of discrete data in the form of lattice points and continuous data from the moduli. The discrete and continuous data can be separated using the \textit{Narain vielbein} $\nviel$. Re-expressing an element of $\mlattn$ as $(p_L,p_R)=\nviel \ell$, where $\ell\in\mathbb{Z}^{2D}$, we see that $p_L^2 - p_R^2 = \ell^T Q_{\rm Narain}\ell$, where $Q_{\rm Narain} = \nviel^T \mathbbm{1}_{D,D}\nviel$ is a quadratic form defining an inner product on the lattice
\begin{equation}
    \Lambda_{\rm Narain} := \left\{\ell\in \mathbb{Z}^{2D} \,\Big|\,  \ell^T Q_{\rm Narain}\ell \in 2\mathbb{Z}  \right\} \;.
\label{eq:narddata}
\end{equation}
This lattice is the discrete data defining the Narain CFT. The moduli, which constitute the continuous data, define $\nviel$ and can packaged in a secondary quadratic form $H=\nviel^T\nviel$ on the lattice, which acts as  $\ell^T H\ell = p_L^2 + p_R^2$. To make this more concrete, we can consider the torus partition function of Narain CFTs,
\begin{equation}
    \lvert\eta(\tau)\rvert^{2D}Z_{\rm Narain}(\tau;m): = \sum_{(p_L,p_R)\in\; \mlattn} q^{p_L^2/2} \bar{q}^{p_R^2/2} =  \sum_{\ell\in \mathbb{Z}^{2D}} \exp\bigg[i\pi\tau_1 Q(\ell) - \pi \tau_2 H(\ell)\bigg]~,
\end{equation}
where $\eta(\tau)$ is the Dedekind eta function, $m$ collectively denotes the CFT moduli and $q= e^{2\pi i \tau}$, with $\tau = \tau_1 + i\tau_2, \ \tau_2 > 0$ the modular parameter of the torus. Thus we see that the discrete and continuous data can be directly used to write the partition function of the Narain CFT. They can also be used to construct the vertex operators in the canonical manner. 

Finally, we come to the Narain moduli space itself, $\modspn$. An arbitrary Narain lattice can be obtained
by the left action of an element of $O(D,D;\mathbb{R})$ on a reference Narain lattice. These lattices define unique Narain CFTs only up to independent $O(D)$ rotations of the left- and right-moving momenta and outer automorphisms of the lattice:
\begin{equation}\label{eq:TDnarain}
    O(D,D;\mathbb{Z}) := \left\{ \Sigma \in \GL(2D,\mathbb{Z})   \,\Big|\, \Sigma^T Q_{\rm Narain}\Sigma = Q_{\rm Narain}   \right\}~,
\end{equation}
which is precisely the T-duality group of the Narain CFT. Then the Narain moduli space is given by the usual double quotient space
\begin{equation}
    \modspn := O(D,D;\mathbb{Z})\backslash O(D,D;\mathbb{R})/ O(D;\R)\times O(D;\R) ~.
\label{eq:narainmod}
\end{equation}
Finally, we note that T-duality manifests as an invariance of the partition function under elements of $O(D,D;\mathbb{Z})$, i.e.
\begin{equation}\label{eq:nparttrans}
    \begin{aligned}
        Z_{\rm Narain}(\tau;g\cdot m) 
        &= Z_{\rm Narain}(\tau;m) \;,
    \end{aligned}
\end{equation}
for all $\Sigma\in O(D,D;\mathbb{Z})$.

\subsection{Defining Data for Generalized Narain Theories}
\label{sec:gennarain}
We now formalize the concepts of the previous subsection and utilize them to construct generalizations of the Narain CFTs. More specifically, we first define the discrete and continuous data (moduli) for a class of theories and use this data to define partition functions and operators of the CFTs. 
In analogy with~\cref{eq:narddata}, we first consider a lattice associated with 
an even, integral, quadratic form $Q$ of signature $(p+, q-)$
\begin{equation}
Q(\ell)=\sum_{i,j=1}^{p+q} Q_{ij}\ell^i \ell^j \in 2\mathbb{Z},  ~~ Q_{ij} \in \text{Mat}_{(p+q) \times (p+q)}(\mathbb Z), ~ \ell \in \mathbb{Z}^{p+q}.
\end{equation}
We shall henceforth refer to an even, integral quadratic form simply as a quadratic form, unless specified otherwise.
A quadratic form corresponds to an even integer lattice, 
which we denote by $\Lambda$:\footnote{An integer lattice is simply a freely generated $\mathbb Z-$module.}
\begin{equation}
\Lambda := \left\{ \ell \in \mathbb{R}^{p+q} \,\Big|\, Q(\ell) \in 2\mathbb{Z} \right\} \;.
\label{eq:discdat}
\end{equation}
Owing to the relation between quadratic forms and integer lattices, one may think of an even quadratic form as a norm function/inner product for vectors in an integer lattice.
In other words, an even quadratic form is bilinear $\displaystyle Q: \Lambda \times \Lambda \rightarrow 2\mathbb Z$ defined as
\begin{equation}
Q(\ell,\ell'):=\frac{Q(\ell+\ell')-Q(\ell)-Q(\ell')}{2}=\sum_{i,j=1}^{p+q} Q_{ij}\ell^i \ell'^j \in 2\mathbb{Z} \;.
\end{equation}
Unlike the previously mentioned self-dual lattices $\Lambda_{\rm Narain}$ of typical Narain CFTs, we will not restrict $\Lambda$ to be self-dual. The dual lattice
\begin{equation}\label{Lambda_dual}
\Lambda^*:= \left\{ x \in \mathbb{R}^{p+q} \,\Big|\, Q(x, \ell) \in  \mathbb{Z} \quad  (\forall \ \ell \in \Lambda) \right\} \;, \ \text{with} \ \Lambda\subset \Lambda^* ~.
\end{equation}
This implies that $\Lambda$ has a non-trivial discriminant group 
\begin{equation}\label{eq:discr}
\disq := \Lambda^* / \Lambda  \;.
\end{equation}
The elements of $\disq$ are equivalence classes $[x]$ such that if $[x]=[y]$ for $x,y\in\Lambda^*$, then $x-y\in\Lambda$. We resort to a mild abuse of notation and use $\alpha$ to refer to both elements of $\disq$ as well as representatives of the equivalence classes. We will specify the distinction in situations where it is important. We will also let $[0]$ denote the equivalence class of the zero element of $\Lambda^*$ so that all representatives of $[0]$ are elements of $\Lambda$. 

For the continuous data, we must define a second quadratic form acting on the elements of $\Lambda$. To define sensible theta functions that we will use to construct partition functions, we must look at elements of the representation space $\mathfrak{h}_Q$ of the quadratic form $Q$ defining $\Lambda$~\cite{Siegel_Lecture}:
\begin{equation}
    \repsp := \left\{H\in \GL(p+q,\mathbb{R}) \,\Big|\, H Q^{-1}H = Q \right\} \;.
\label{eq:repsp}
\end{equation}
One could solve the condition on $H$ in~\cref{eq:repsp} and parameterize the solutions --- the resulting expression for $H$ will depend on $pq$ number of parameters, giving dim$_{\mathbb{R}}\repsp = pq$. This is equivalent to diagonalizing the quadratic form $Q$ via the Narain vielbein\footnote{where ``viel" $=p+q$.} $\nviel$ as $Q = \nviel^T \mathbbm{1}_{p,q}\nviel$\footnote{The signature of the diagonalized quadratic form remains the same due to Sylvester's law of inertia~\cite{sylvester1852xix}.}  
and defining $H=\nviel^T\nviel$, as described in \cref{sec:stdnarain}. Clearly, two distinct vielbeins $\nviel$ and $\nviel^\prime$ are related by an element $\mathfrak{R}\in \rmO(p,q;\mathbb{R})$ as $\nviel^\prime = \mathfrak{R}\nviel$ and so there exists $H,H^\prime\in \repsp$ such that $H^\prime = \mathfrak{R}^T H \mathfrak{R}$.

To properly define the continuous data of the generalized Narain CFTs, we must describe the moduli space $\mathcal{M}_\Lambda$ that defines distinct quadratic forms $H$ and thereby distinct CFTs. Due to the existence of a non-trivial discriminant group, this is a subtle point we return to below. For now, we simply assume the existence of this moduli space and state that distinct Narain CFTs correspond to points $m\in\mathcal{M}_{\Lambda}$. 

With the discrete and continuous data $\Lambda$ and $\repsp$ in hand, we can now define the \textit{generalized Narain lattice} $\mlatt$ associated to $\Lambda$ (and hence $Q$) as
\begin{equation}\label{eq:genmom}
    \mlatt := \left\{ \mathfrak{p}=(p_L, p_R) = \mathcal{E}\ell \,\Big|\, \ell\in\Lambda\;\; \&\;\; H = \nviel^T\nviel \in\repsp \right\} \;.
\end{equation}
The inner product on $\mathcal{P}_{\Lambda}$ is induced from that on $Q(\ell)$ as
\begin{equation}
    \begin{aligned}
        \mathbb{I}(\mathfrak{p}) &= \mathfrak{p}^T\mathbbm{1}_{p,q}\mathfrak{p} = p_L^2 - p_R^2
                        = \ell^T Q \ell \in 2\mathbb{Z}
    \end{aligned}
\end{equation}
for all $\mathfrak{p} = \nviel\ell\in\mlatt$. Thus we see clearly that the introduction of the Narain vielbien allows a decomposition into left- and right-movers in analogy with standard Narain CFTs. Indeed we also see that $H(\ell) = p_L^2 + p_R^2$. A convenient way to encapsulate the decomposition of the momenta is via the quantities 
\begin{equation}
    Q_L(\ell):=  \frac{Q}{2}+\frac{H}{2} \;, \quad Q_R(\ell):=  \frac{H}{2} - \frac{Q}{2} \;,
\label{eq:Qdecomp}
\end{equation}
such that $Q_L(\ell) = p_L^2$ and $Q_R(\ell) = p_R^2$.

We also define the dual Narain lattice 
\begin{equation}\label{eq:genmomdual}
    \mathcal{P}_\Lambda^* := \left\{ \rho = (\rho_L, \rho_R) = \mathcal{E} \alpha \,\Big|\,  \alpha\in\Lambda^*\;\; \&\;\; H := \nviel^T\nviel \in\repsp \right\} \;,
\end{equation}
whose elements satisfy the expected property $\mathbb{I}(\rho,\mathfrak{p}) = \rho^T \mathbbm{1}_{p,q}\mathfrak{p} = x^T Q\ell \in\mathbb{Z}$ for all $\rho=\nviel \alpha\in\mlatt^*$ and $\mathfrak{p}=\nviel\ell\in\mlatt$. 
Let us also define a subset of $\mathcal{P}_\Lambda^*$ associated with a particular element $[\alpha]$ of $\mathscr{D}_{\Lambda}$:
\begin{equation}
    \mathcal{P}^*_{\Lambda, [\alpha]} := \left\{ \rho = (\rho_L, \rho_R) = \mathcal{E} \beta \,\Big|\,  \beta\in\Lambda^*\;\; \&\;\;  [\alpha]=[\beta] \in \mathscr{D}_{\Lambda} \;\; \&\;\;  H := \nviel^T\nviel \in\repsp \right\} \;,
\end{equation}
Note that this is not a lattice in general, and satisfies 
$\mathcal{P}_{\Lambda^{*}, [\alpha]}  \cdot \mathcal{P}_{\Lambda^{*}, [\beta]}
\subset \mathcal{P}_{\Lambda^{*}, [\alpha+\beta]}$.

\begin{figure}
\centering
\begin{tikzpicture}
\node at (0,3.8) {\scalebox{1}{\textcolor{black}{Discrete Data: }$\textcolor{DESYO}{\Lambda}\Rightarrow \textcolor{DESYC}{\disq}$}};
\node at (6.5,3.8) {Continuous Data: \textcolor{DESYB}{$\modsp$}};
\node at (5.9,1.6) {$m$};
\node at (-1,-2.3) {Vertex Operators};
\node at (-1,-2.9) {$\mathcal{V}^{[\alpha]}_{(k_L,k_R)}$};
\node at (2.9,-2.3) {Hilbert Space};
\node at (3.1,-2.9) {$\mathcal{H}_\Lambda$};
\node at (6.9,-2.3) {Partition Functions};
\node at (6.9,-2.9) {$Z^\Lambda_\alpha$};

{
\begin{scope}
\node[] (L1) at (1.4,2) {}; 
\node[] (L2) at (0,0.5) {}; 

\node[] (M1) at (4.5,2) {}; 
\node[] (M2) at (5.4,0.5) {}; 

\node[] (N1) at  (0.7,-0.7)  {};
\node[] (N2) at  (5,-0.7)  {};

\node[] (V1) at (0.1,-0.9) {};
\node[] (V2) at (-1.1,-2.1) {};
\node[] (V3) at (0.5,-2.3) {};

\node[] (H1) at (3,-1.1) {};
\node[] (H2) at (3,-2.1) {}; 
\node[] (H3) at (1.8,-2.3) {};
\node[] (H4) at (4.1,-2.3) {};

\node[] (P1) at (5.7,-0.9) {};
\node[] (P2) at (7,-2.1) {};
\node[] (P3) at (5.2,-2.3) {};

\draw (L1) edge [->,thick]  (M1);
\draw (M2) edge [->,thick]  (N2);
\draw (L2) edge [->,thick]  (N1);
\draw (V1) edge [->,thick]  (V2);
\draw (H1) edge [->,thick]  (H2);
\draw (P1) edge [->,thick]  (P2);
\draw (V3) edge [->,thick]  (H3);
\draw (H4) edge [->,thick]  (P3);
\end{scope}
}

{
\begin{scope} [scale=1,thick,rotate = 45]
\draw[step=0.7cm] (0.1,0.1) grid (2.75,2.75);
\foreach \p in {(0.7,0.7),(1.4,0.7),(2.1,0.7)}
\fill[DESYO] \p circle(.1);
\foreach \p in {(0.7,1.4),(1.4,1.4),(2.1,1.4)}
\fill[DESYO] \p circle(.1);
\foreach \p in {(0.7,2.1),(1.4,2.1),(2.1,2.1)}
\fill[DESYO] \p circle(.1);

\foreach \p in {(1.4,1.4), (1.4-0.35,1.4),(1.4-0.35,1.4+0.35),(1.4-0.35,1.4-0.35),(1.4+0.35,1.4),(1.4+0.35,1.4),(1.4+0.35,1.4-0.35),(1.4+0.35,1.4+0.35),(1.4,1.4+0.35),(1.4,1.4-0.35)}
\fill[DESYC] \p circle(.07);
\draw[line width=0.7pt,color=DESYC] (1.4-0.35,1.4+0.35) -- (1.4+0.35,1.4+0.35) -- (1.4+0.35,1.4-0.35) -- (1.4-0.35,1.4-0.35) -- (1.4-0.35,1.4+0.35);
\end{scope}
}

{
\begin{scope} [shift={(4.7,0.5)},scale=1,rotate around={58:(1.4,1.4)}]
\foreach \p in {(0.8,1.6)}
\fill[DESYB] \p circle(.1);

\draw[thick] (0.7,2.8) to [out=-60,in=-120]  (2.1,2.8)
 to [out=-90, in=180] (2.8,2.1) to [out=210, in=150] (2.8,0.7) to [out=180, in=90 ] (2.1,0) to [out=120, in=60]  (0.7,0) to [out=90, in=0] (0,0.7) to [out=30, in=-30] (0,2.1) to [out=0, in=-90] (0.7,2.8);
\end{scope}
}

{
\begin{scope} [shift= {(2.8,-1)}, scale=0.7]
{\node[] (N) at (0,0) {\scalebox{1}{ Generalized Lattices $\mlatt$ and $\mlatt^*$}};}
\end{scope}
}
\end{tikzpicture}
\caption{Logical flow of the construction of the generalized CFTs described in this section. A choice of integral lattice $\textcolor{DESYO}{\Lambda}$ (\cref{eq:discdat}) defines a discriminant group $\textcolor{DESYC}{\disq}$ (\cref{eq:discr}) as well as the moduli space $\textcolor{DESYB}{\modsp}$ (\cref{eq:genmod}) of the CFT. This data, with a choice of $m\in\textcolor{DESYB}{\modsp}$, can then be used to define the generalized Narain lattice $\mlatt$ (\cref{eq:genmom}) and its dual $\mlatt^*$ (\cref{eq:genmomdual}), vertex operators (~\cref{eq:vertexop}), the Hilbert space $\mathcal{H}_\Lambda$, and partition functions $Z^{\Lambda}_{\alpha}(\tau, \taubar;m)$ (~\cref{ZQh}).}
\end{figure}
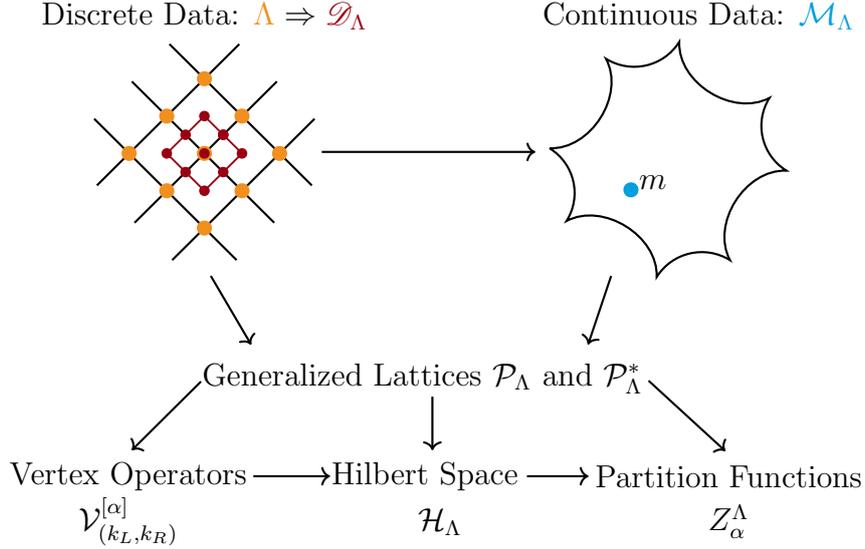

Similar to standard Narain CFTs, the primary operators of our generalized theories are built from compact chiral bosons with periodicity conditions
\begin{equation}
    \begin{pmatrix}
        X^a_L\\ X^{\mbar}_R
    \end{pmatrix} \sim \begin{pmatrix}
        X^a_L\\ X^{\mbar}_R
    \end{pmatrix} + \nviel L 
\end{equation}
with $L\in\mathbb{Z}^{p+q}$, $a=1,2,\dots, p$ and $\mbar= 1,\dots ,q$. The operators consist of holomorphic and anti-holomorphic currents  
\begin{equation}
    \begin{aligned}
        J^a(z) &:= i\partial X^a_L(z) \;,\\
        \jbar^{\mbar}(\zbar) &:= -i\bar{\partial} X^{\mbar}_R(\bar{z}) \;,
    \end{aligned}
\end{equation}
and vertex operators 
\begin{equation}\label{eq:vertexop}
   \mathcal{V}^{[\alpha]}_{(k_L,k_R)}(z,\zbar) :=\;\; :\exp\bigg(    i k_L\cdot X_L(z) + i k_R\cdot X_R(\zbar)\bigg): \;,
\end{equation}
where $: \, :$ denotes the normal ordering, $(k_L^a,k_R^{\mbar})\in \mathcal{P}^*_{\Lambda, [\alpha]}$, and the dot products are simply $k_L\cdot X_L = \sum_{a=1}^pk_L^a X^a_L$ and $k_R\cdot X_R = \sum_{m=1}^q k^m_R X_R^{\mbar}$. We also can define holomorphic and anti-holomorphic stress tensors
\begin{equation}
    \begin{aligned}
        T(z) &:=\;\; \sum_a:J^a(z)J^a(z): \;,\\
        \bar{T}(\zbar) &:=\;\; \sum_{\mbar}:\jbar^{\mbar}(\zbar)\jbar^{\mbar}(\zbar): \;,
    \end{aligned}
\end{equation}

The currents have the standard operator product expansions (OPEs)
\begin{equation}
    J^a(z)J^b(0) = \frac{\delta^{ab}}{2z^2}+\cdots \quad  \jbar^{\mbar}(\zbar)\jbar^{\bar{n}}(0)
    = \frac{\delta^{\mbar\bar{n}}}{2\zbar^2}+ \cdots  \;,
\end{equation}
which give the vertex operator OPE
\begin{equation}
    \mathcal{V}^{[\alpha_1]}_{(k_L,k_R)}(z,\zbar) \mathcal{V}^{[\alpha_2]}_{(k^\prime_L,k^\prime_R)}(0,0) = z^{k_L\cdot k_L^\prime}\zbar^{k_R\cdot k^\prime_R}\mathcal{V}^{[\alpha_1+\alpha_2]}_{(k_L+k^\prime_L,k_R+k^\prime_R)}(0,0)+\cdots \;.
\end{equation}
The existence of $\mathcal{P}_{\Lambda}$ guarantees the closure of this OPE.
Note while the vertex operators $\mathcal{V}^{[\alpha]}_{(k_L,k_R)}$
are mutually local for $\alpha=0$, general vertex operators $\mathcal{V}^{[\alpha]}_{(k_L,k_R)}$ with $\alpha\ne 0$ are not necessarily mutually local.

The Virasoro algebra can be realized from the mode expansion of the stress-energy tensor in the usual way, and we denote the operators by $L_h$ and $\widetilde{L}_h$.

We can define states in the theory using the typical state-operator correspondence. Apart from oscillator modes, we have momentum modes
\begin{equation}\label{eq:momstates}
    \ket{\alpha; k_L,k_R} = \lim_{z,\zbar\rightarrow 0} \mathcal{V}^{[\alpha]}_{(k_L,k_R)}(z,\zbar)\ket{0} \;,
\end{equation}
which satisfy the usual relations
\begin{equation}
    L_0\ket{\alpha; k_L,k_R} = \frac{k_L^2}{4}\ket{\alpha; k_L,k_R} \;, 
    \qquad  
    \widetilde{L}_0\ket{\alpha; k_L,k_R} = \frac{k_R^2}{4}\ket{\alpha; k_L,k_R} \;.
\end{equation}
One can also define higher oscillator states in the usual way, and the combination of momenta and oscillator states defines the Hilbert space $\mathcal{H}_\Lambda$ of the theory.  

From the above, we have obtained definitions for primary operators, the stress-energy tensor, sensible and closed OPEs, and states in the theory. 
It is also clear that the above is a unitary CFT with the typical inner product between the momenta states defined in~\cref{eq:momstates}. One may also wonder about the modular invariance of the partition function. We define the partition functions in the next section and return to this question there.

Finally, one might be interested in the question of formulating the existence of generalized CFTs in a mathematically rigorous fashion.
For this purpose, one needs to choose a mathematical definition of CFT\footnote{Discussions on toroidal CFTs can also be found in \cite{Wendland:2000ye, Kidambi:2022wvh, Kapustin:2000aa, Kontsevich:2000yf}.} \footnote{While the Vertex Operator Algebras (VOAs)
have often been used in the mathematical formulations of CFT, this applies only to the chiral (anti-chiral) part of the CFTs, and is insufficient for our purposes here.}.
One attempt in this direction is \cite{Moriwaki:2020cxf}, which defined ``full VOAs'' for CFTs on the plane
and showed that generalized Narain CFTs associated with a lattice 
satisfy the axioms therein. The mathematical existence of generalized Narain CFTs is proven in this sense.

\subsection{Torus Partition Functions}

The torus partition function is given by the trace over the Hilbert space $\mathcal{H}_{\Lambda}$  of the theory on the spatial $S^1$:
\begin{align}
Z^\Lambda(\tau, \taubar;m) =\textrm{Tr}_{\mathcal{H}_{\Lambda}} \left(q^{L_0-\frac{c}{24}} \qbar^{\overline{L}_0-\frac{\cbar}{24}} \right) =
\frac{1}{\eta(\tau)^p \etabar(\taubar)^q} \sum_{\ell \in \Lambda} e^{i \pi \tau Q_L(\ell)- i \pi \taubar Q_R(\ell)} \;,
\label{ZQ}
\end{align}
where $q := \exp(2\pi i \tau)$ and $\qbar := \exp(-2\pi i \taubar)$, $Q_L$ and $Q_R$ are defined in terms of the Narain data as in~\cref{eq:Qdecomp}, and we have explicitly shown the dependence on the moduli $m \in \mathcal{M}_{\Lambda}$.

The modular transformations of the partition function $Z^{\Lambda}$ is better
described in terms of building blocks
\begin{align}
Z^{\Lambda}_{\alpha}(\tau, \taubar;m) =
\frac{\vartheta^{\Lambda}_{\alpha}(\tau, \taubar;m)}{\eta(\tau)^p \etabar(\taubar)^q} \;,
\label{ZQh}
\end{align}
which serve as the basis for more general partition functions.

Here the theta function $\vartheta^{\Lambda}_{\alpha}$ is given by 
\begin{equation}\label{eq:gentheta}
\begin{aligned}
    \vartheta^{\Lambda}_{\alpha}(\tau, \taubar;m) 
    &:= \sum_{\ell \in {\Lambda+\alpha}} e^{i \pi \tau Q_L(\ell)- i \pi \taubar Q_R(\ell)} 
    = \sum_{\ell \in {\Lambda+\alpha}} e^{i \pi \tau_1 Q(\ell)-  \pi \tau_2 H(\ell)}\;.
\end{aligned}
\end{equation}
Here $\alpha$ is a representative of $[\alpha]\in\disq$. There are $\lvert\disq\rvert$ theta functions, which are defined only by the equivalence class and not by the particular representatives chosen, as can be seen by re-defining the sum over $\Lambda$. The partition function \eqref{ZQh} reduces to \eqref{ZQ} when $\alpha =0 \in \disq$.
By relabeling the sum in \eqref{ZQh} as $\ell \to -\ell$, we obtain
\begin{align}\label{theta_flip}
\vartheta^{\Lambda}_{\alpha}(\tau, \taubar;m) = \vartheta^{\Lambda}_{-\alpha}(\tau, \taubar;m)\;.
\end{align}
The generators of $\SL(2, \mathbb{Z})$ are given by the matrices $T$ and $S$:
\begin{align}
T= \left(
\begin{array}{cc}
1 & 1\\
0 & 1 
\end{array}
\right)
\;, 
\quad 
S= \left(
       \begin{array}{cc}
          0 & -1\\
          1 & 0 
        \end{array}
    \right)
\;.
\end{align}
We have relations $(ST)^3=1, S^2=- 1$. 
The corresponding modular transformations are\footnote{Strictly speaking we should rather consider the double cover of $\SL(2, \mathbb{Z})$, the metaplectic group $\textrm{Mp}(2, \mathbb{Z})$, and the representation is in \eqref{eq:Weil} is the Weil representation of $\textrm{Mp}(2, \mathbb{Z})$ \cite{MR0165033}.}
\begin{equation} \label{eq:Weil}
\begin{split}
&T: \quad \vartheta^{\Lambda}_{\alpha }(\tau+1; m) 
        =  e^{i\pi  Q(\alpha  )}  \vartheta^{\Lambda}_{\alpha }(\tau; m) \;,\\
&S: \quad \vartheta^{\Lambda}_{\alpha }\left(-\frac{1}{\tau};  m\right) 
        =\frac{1 }{\sqrt{|\textrm{det}\, Q|} } \, e^{-\frac{\pi i}{4}(p-q)} \tau^{\frac{p}{2}} \taubar^{\frac{q}{2}} 
            \sum_{\gamma \in  \mathscr{D}_{\Lambda}}  e^{-2\pi i Q(\alpha , \gamma)}  \vartheta^{\Lambda}_{\gamma}(\tau; m) \;,
\end{split}
\end{equation}
where here and in the following we often drop the $\taubar$ dependence from the notation. 

Let us consider a modular transformation
\begin{align}\label{M_def}
\tau \to \tau_{[M]} := \frac{a\tau+b}{c\tau +d} \;, \quad
[M]=\pm \left(\begin{array}{cc} a & b \\ c& d \end{array} \right) \in \PSL(2, \mathbb{Z}) 
\quad
(c\ne 0) \;,
\end{align}
where in the following we will denote an element of $\PSL(2, \mathbb{Z})$
as $[M] = \pm M$, with $M$ being an $\SL(2, \mathbb{Z})$ representative.
Under the modular transformation \eqref{M_def}
the theta functions mix among themselves:
\begin{align}\label{theta_Lambda_SL2}
     \vartheta^{\Lambda}_{\alpha}(\tau_{[M]};m)
     & = \sum_{\beta \in \mathscr{D}_\Lambda}\mathcal{U}^{\Lambda}_{\alpha, \beta}(M, \tau) \,\vartheta^{\Lambda}_{\beta}(\tau;m) \;.
 \end{align}
Here the transformation matrix $\mathcal{U}$ is given by
\begin{align} \label{Ucal_def}
     \mathcal{U}^{\Lambda}_{\alpha, \beta}(M, \tau) := 
     \begin{cases}
         \displaystyle   (c\tau+d)^{\frac{p}{2}} (c\taubar+d)^{\frac{q}{2}}  \lambda^{\Lambda}_{\alpha, \beta}(M)  &  (c\ne 0)\;, \\
         \displaystyle   \mathcal{U}^{\Lambda}_{\alpha, \beta}(T, \tau) := e^{i\pi  Q(\alpha  )} \delta_{\alpha, \beta} & (c=0) \;,
     \end{cases}
\end{align}
and a version of the quadratic Gauss sum by
 \begin{align} \label{lambda_def}
     \lambda^{\Lambda}_{\alpha, \beta}(M):= \frac{1}{\sqrt{|\textrm{det}\, Q|} } \, e^{-\frac{\pi i}{4}(p-q)} c^{-\frac{p+q}{2}} \sum_{\ell_c \in \Lambda/( c  \Lambda)} e^{\frac{\pi i}{ c  }\left(a Q[\ell_c +\alpha] -2Q[\ell_c +\alpha, \beta]+ d   Q[\beta] \right)} \;.
\end{align}
It is easy to verify that $\lambda_{\alpha, \beta}$ is preserved when we shift $\alpha, \beta$ by elements of $\Lambda$, so that
we can regard $\alpha, \beta$ to be elements of the discriminant $\mathscr{D}_{\Lambda}$, as implied by the notation.

Note that $\mathcal{U}^{\Lambda}_{\alpha, \beta}(M, \tau)$ and $\lambda^{\Lambda}_{\alpha, \beta}(M)$ both depend on the 
choice of the $\SL(2, \mathbb{Z})$ representative $M$. 
We can also verify, however, that the relation \eqref{theta_Lambda_SL2} is not affected by this choice; 
this is verified by the relations \eqref{theta_flip}, $\mathcal{U}_{\alpha, \beta}(M) = \mathcal{U}_{\alpha, -\beta}(-M)$ and then a relabeling of the sum as $\beta\to -\beta$. In the following we will denote both $[M]$ and $M$ simply by $M$,
to avoid clutter in notations.

The modular transformation formula \eqref{theta_Lambda_SL2} was known long ago to C.~Siegel \cite[\S 3]{Siegel_Lecture}.
The matrix $\mathcal{U}$ satisfies
\begin{align}\label{M_consistent}
     \mathcal{U}^{\Lambda}_{\alpha, \beta}(M \cdot M', \tau)  =    \sum_{\gamma \in \mathscr{D}_{\Lambda}}  \mathcal{U}^{\Lambda}_{\alpha, \gamma}(M, \tau_{M'})  \,\mathcal{U}^{\Lambda}_{\gamma, \beta}(M', \tau) \;,
\end{align}
which ensures the consistency of the modular transformation \eqref{theta_Lambda_SL2} under the composition of $\SL(2, \mathbb{Z})$ elements.
Note that $\lambda$ themselves satisfy $\tau$-independent relations
\begin{align}\label{lambda_consistent}
     \lambda^{\Lambda}_{\alpha, \beta}(M \cdot M')  =    \sum_{\gamma \in \mathscr{D}_{\Lambda}}  \lambda^{\Lambda}_{\alpha, \gamma}(M)  \,\lambda^{\Lambda}_{\gamma, \beta}(M') \;.
\end{align}

The modular transformations of the partition functions \eqref{ZQh} are given by
\begin{equation}\label{CFT_Z_TS}
\begin{split}
&T: \quad Z^\Lambda_{\alpha }(\tau+1; m) = e^{-\frac{2\pi i (p-q)}{24}} e^{i\pi  Q(\alpha  )}  Z^\Lambda_{\alpha }(\tau; m) \;,\\
&S: \quad Z^\Lambda_{\alpha }\left(-\frac{1}{\tau};  m\right) =
\frac{1}{\sqrt{\vert \textrm{det}\, Q}\vert }  \sum_{\gamma \in  \mathscr{D}_{\Lambda}}
e^{-2\pi i Q(\alpha , \gamma)}  Z^\Lambda_{\gamma}(\tau; m) \;,
\end{split}
\end{equation}

We comment here on the role of modular invariance and the generalized CFTs defined above. The consistency conditions for a $2d$ CFT on arbitrary $2d$ surfaces include crossing symmetry of the sphere four-point function and modular invariance of the torus partition function and torus one-point function~\cite{Moore:1988qv}. From the transformation properties in~\cref{CFT_Z_TS}, it is clear that a given torus partition function $Z^\Lambda_\alpha$ is not invariant under $\PSL(2,\mathbb{Z})$ since the elements of the discriminant group transform into one another under the $S$ transformation. (The exception to this statement is if the quadratic form $Q$ corresponds to a self-dual lattice and thus the discriminant group is trivial.)
One possible viewpoint is that our CFTs are limited and not defined on arbitrary $2d$ surfaces.
Alternatively, we can consider the CFTs defined above as 
being a subsector of a CFT with expanded content such that the total partition function is modular invariant.\footnote{For example, a modular invariant theory could be constructed by combining a generalized Narain CFT defined by quadratic form $Q$ with a ``conjugate" CFT that is defined by $-Q$.} Such a sector would furnish a projective representation of  $\PSL(2,\mathbb{Z})$ similar to that defined in~\cref{CFT_Z_TS}, although the theory in itself need not be identical to the generalized Narain CFT sector. There is some indication that such an expansion is necessary. For example, the bulk Chern-Simons dual of the generalized Narain CFTs contains a global one-form symmetry that should not exist within the context of quantum gravity. Thus at a bare minimum the content of the expanded CFT is such that this global symmetry is eliminated. While~\cite{Gukov:2004id} discusses a stringy origin of the bulk Chern-Simons theories, we expect that the discussion there should be supplemented with additional ingredients based on the arguments here.   

\subsection{Operators and Moduli Space}
Similar to standard Narain CFTs, the primary operators of our generalized theories are holomorphic and anti-holomorphic currents formed from compact bosons 
\begin{equation}
    \begin{aligned}
        J^a(z) &:= \partial X^a \;, \\
        \bar{J}^{\mbar}(\bar{z}) &:= \bar{\partial} X^{\mbar} \;,
    \end{aligned}
\end{equation}
and vertex operators 
\begin{equation}
   \mathcal{V} :=\;\; :\exp\bigg(    i p_L\cdot X_L + i p_R\cdot X_R\bigg): \;,
\end{equation}
where $a=1,\dots,p$ and $\mbar=1, \dots, q$, and $: \, :$ denotes the normal ordering. The existence of $\mathcal{P}_{\Lambda}$ 
guarantees the closure of the associated Operator Product Expansion (OPE). 

We now return to the issue of the moduli space of the generalized CFTs. Naively, one may draw an analogy with the T-duality group of Narain CFTs in~\cref{eq:TDnarain} and define the T-duality group for the generalized theories by simply collecting integral transformations that preserve the quadratic form $Q$:
\begin{equation}\label{eq:genTD}
    \mathrm{O}_Q(p,q;\mathbb{Z}) := \left\{  \Sigma \in \GL(p,q;\mathbb{Z})    \,\Big|\,  \Sigma^T Q \Sigma = Q  \right\} \;.
\end{equation}      
This is essentially correct, but there are several subtleties that must be addressed. First, the generalized partition functions in~\cref{ZQh} are not necessarily invariant under the elements of~\cref{eq:genTD}, and we instead have \cite{Ashwinkumar:2023jtz}
\begin{equation}\label{eq:genthetaTD}
    \vartheta^\Lambda_{g\cdot \alpha}(\tau,\taub;g\cdot m) = \vartheta^\Lambda_\alpha (\tau,\taub;m) \;,
\end{equation}
where $g\cdot\alpha := \Sigma \alpha$ and $g\cdot m$ implies the replacement $H\rightarrow \Sigma^TH\Sigma$ for $\Sigma\in O_Q(p,q;\mathbb{Z})$. 

One is thus led to define T-duality groups that not only preserve $Q$ but also an element $\alpha$ the discriminant group:
\begin{equation}\label{eq:genTDa}
   \mathrm{O}_{Q,\alpha}(p,q;\mathbb{Z})  := \left\{ \Sigma \in \GL(p+q;\mathbb{Z})         \,\Big|\,  \Sigma^T Q \Sigma = Q \; \& \; [\Sigma\alpha] = [\alpha] \in\disq         \right\} \;.
\end{equation}
Note that $\displaystyle \rmO_{Q,\alpha}(p,q;\mathbb{Z}) \subseteq \rmO_Q(p, q;\mathbb{Z})$. Then we can define a moduli space for each element of the discriminant group as  
\begin{equation} \label{eq:genmoddis}
\mathcal{M}_{Q,\alpha} :=
     \mathrm{O}_{Q,\alpha}(p,q;\mathbb{Z}) \big\backslash  \widetilde{\mathcal{M}}_{Q}\;, \quad
\widetilde{\mathcal{M}}_{Q} 
      := (\rmO(p,q; \mathbb{R})\big/(\rmO(p; \mathbb{R})\times \rmO(q; \mathbb{R})) \fstop
\end{equation}
which we sometimes denote as $\mathcal{M}_{\Lambda,\alpha}$.
However, these spaces are not the moduli space of the complete theory, which includes all of the elements of $\disq$. To determine this, we utilize several more facts about the T-duality group and $\disq$ (see \Cref{appendix:TDproofs} for proofs).

The elements of $\genTD$ define bijective maps from $\disq$ to itself that leave $[0]$ invariant. Thus the T-duality group associated to $[0]$ is $\mathrm{O}_{Q,[0]}(p,q;\mathbb{Z}) = \genTD$. For the other equivalence classes in $\disq$, the elements of $\genTD$ can either i) leave the class invariant, as for the case of $[0]$ or ii) map different equivalence classes into each other. Classes that satisfy the former property again have T-duality groups given by~\cref{eq:genTD}. For classes that satisfy the latter property, we associate with them the more limited groups of~\cref{eq:genTDa}. However, so long as we include all the different equivalence classes in the definition of the theory, the T-duality group of the total theory is the full one presented in~\cref{eq:genTD}. These points are illustrated in~\Cref{fig:tdualpic}.

Thus in the full theory that contains all of the elements of $\disq$, the true moduli space is  
\begin{equation} \label{eq:genmod}
\mathcal{M}_{Q} :=
     \mathrm{O}_{Q}(p,q;\mathbb{Z}) \big\backslash  \widetilde{\mathcal{M}}_{Q}\fstop
\end{equation}
which we sometimes denote as $\mathcal{M}_{\Lambda}$.

The moduli spaces \eqref{eq:genmoddis} are endowed with a metric, known as the \textit{Zamolodchikov metric} whose measure we denote by $[dm]$. This measure is the Haar measure of an orthogonal symmetric space with a canonical normalization.
The dimension of the moduli space \eqref{eq:genmoddis} is $\textrm{dim}_{\mathbb{R}}\, \mathcal{M}_{Q,\alpha}=  p q$. 

\begin{figure}[ht!]
\centering
\begin{tikzpicture}
{
\begin{scope}
\draw[thick,fill=DESYB!30] (1,0) to [out=25,in=-120+55] (2,2) to [out=-45,in=-120+45]  (6,4) to [out=-90+25,in=180] (8,2.8) to (8,0) ;
\end{scope}
}

{
\begin{scope}
\draw[thick,fill=DESYO!30]  (2,2) to [out=-15,in=-110] (6,4) to [out=-120+45,in=-45] (2,2);
\end{scope}
}
{
\begin{scope}
\draw[thick,fill=DESYO!30]  (6,4) to [out=-90+25,in=200] (8,3.4) to (8,2.8) to  [out=180,in=-90+25](6,4);
\end{scope}
}

{
\begin{scope}
\draw[thick,fill=DESYO!30] (1,0) to [out=50,in=-120+25] (2,2) to [out=-120+55,in=25] (1,0);
\end{scope}
}

{
\begin{scope}
\node at (3.5,2.2) {$[\alpha_2]$};
\node at (4.8,0.9) {$[\alpha_3]$};
\node at (7,0.5) {$\modsp$};
\node at (1,1.3) {$[\alpha_1]$};
\node at (2.8,0.6) {$[\alpha_1]$};
\node at (7.4,3.7) {$\mathcal{M}_{\Lambda,\alpha_2}$};
\draw[thick,->] (7.2,3.5) -- (7.7,3) ;
\end{scope}
}

{
\begin{scope}[shift={(2,-5.1)},rotate around={35:(2,6.5)}]
\draw[thick,rounded corners] (1.93,7) to [out=-90,in=120] (2,6.5);
\draw[thick,rounded corners] (2,6.5) to [out=60+180,in=90] (1.93,6.1);

\draw[thick,rounded corners] (2.07,7) to [out=-90,in=60] (2,6.5);
\draw[thick,rounded corners] (2,6.5) to [out=120-180,in=90] (2.05,6.1);

\draw[thick] (1.85,6.2) to (2,6.03) to (2.15,6.2);
\foreach \p in {(2,7), (2,5.9)}
\fill[DESYC] \p circle(.1);
\end{scope}
}

{
\begin{scope}[shift={(-0.1,-5.5)},rotate around={50:(2,6.5)}]
\draw[thick,rounded corners] (1.93,7) to (1.93,6.1);

\draw[thick,rounded corners] (2.07,7) to (2.07,6.1);

\draw[thick] (1.85,6.2) to (2,6.03) to (2.15,6.2);
\foreach \p in {(2,7), (2,5.9)}
\fill[DESYC] \p circle(.1);
\end{scope}
}

\draw[thick] (8,0) to (0,0) to (0,5) to (8,5) to (8,3.4);

\end{tikzpicture}
\caption{Illustration of T-duality on elements of $\disq$ in the space of moduli values of the generalized CFT. The \textcolor{DESYB}{blue} region is the moduli space of the generalized CFT, $\modsp$, the \textcolor{DESYO}{orange} region is the moduli space of $[\alpha_2]\in\disq$, $\mathcal{M}_{\Lambda, \alpha_2}$, and the white space corresponds to moduli values outside the fundamental domain. The two possibilities outlined in the text are i) an element of $[\alpha_1]\in\disq$ is preserved by all elements of $\genTD$ and so is unchanged when mapping into $\modsp$ or ii) an element $[\alpha_1]\in\disq$ is mapped to $[\alpha_2]\neq [\alpha_1]$ under $\genTD$. }
\label{fig:tdualpic}
\end{figure}
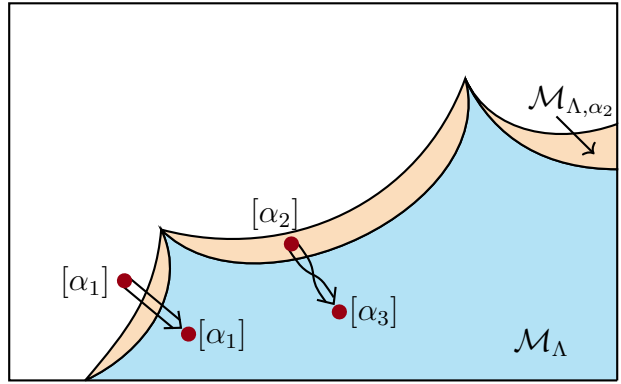

We can illustrate the points above by a simple example
\begin{equation}
    Q = \begin{pmatrix}
        0 & 2\\
        2 & 0
    \end{pmatrix}\fstop
\end{equation}
Then $|\text{det}(Q)| = \lvert\disq\rvert = 4$ and we can choose representatives of the elements of $\disq$ as
\begin{equation}\label{eq:exanyons}
    \begin{aligned}
    \alpha_0 &= \left(0,0\right) \;,\quad
    \alpha_1 &= \left(\frac{1}{2},\frac{1}{2}\right) \;,\quad
    \alpha_2 &= \left(\frac{1}{2},0\right)\;,\quad
    \alpha_3 &= \left(0,\frac{1}{2}\right)\;.
    \end{aligned}
\end{equation}
The vielbein is parameterized by elements of $\mathrm{O}(1,1;\mathbb{R})$ as 
\begin{equation}
    \mathcal{E} = \begin{pmatrix}
                    \cosh\phi & \sinh\phi\\
                    \sinh\phi & \cosh\phi
                    \end{pmatrix}
                    \begin{pmatrix}
                        1 & 1\\
                        -1 & 1
                    \end{pmatrix} 
                    \;,
\end{equation}
which gives a parameterization for the elements of $\repsp$:
\begin{equation}
    H = \begin{pmatrix}
        2(\cosh(2\phi) - \sinh(2\phi)) & 0\\
        0 & 2(\cosh\phi + \sinh\phi)^2
        \end{pmatrix}=: \begin{pmatrix}
            A &0 \\
            0 & 4A^{-1}
        \end{pmatrix}\fstop
\end{equation}
Finally, the T-duality group consists of a single element other than the identity, up to an overall sign:
\begin{equation}
    \Sigma = \begin{pmatrix}
                0 &  1\\
                 1 & 0
    \end{pmatrix}
    \fstop
\end{equation}
The action $H\rightarrow \Sigma^T H \Sigma$ of $\Sigma$ is equivalent to the replacement $A\rightarrow 4A^{-1}$. Then the naive moduli space is the interval $[4,\infty)$ (or equivalently, the interval $(0,4]$). However, we first note the action of $\Sigma$ on elements of $\disq$. Using the representatives in~\cref{eq:exanyons}, we see that $[\alpha_0]$ and $[\alpha_1]$ are invariant but
\begin{align} 
    \begin{split}
        [\Sigma\alpha_2] &= [\alpha_3]\;, \quad
        [\Sigma\alpha_3] = [\alpha_2]
    \end{split}
\end{align}
Thus the moduli spaces of the sectors of the theory described by $\alpha_{0,1}$ are identical to the naive one: $\mathcal{M}_{\Lambda, \alpha_0} = \mathcal{M}_{\Lambda,\alpha_1} = [4,\infty)$. For the remaining sectors, we have
$\mathcal{M}_{\Lambda,\alpha_2} = \mathcal{M}_{\Lambda,\alpha_3} = (0,\infty)$. However, if we consider any point $m\in (0,4)$ for either of these latter two classes, we see that it can be mapped to the interval $[4,\infty)$, provided that we swap the sector simultaneously
as $\alpha_2\leftrightarrow \alpha_3$. Thus we see that the moduli space of the full theory is indeed the naive one, in agreement with the discussion above. 

\section{Orbifolded Narain CFTs}\label{sec:orbifold}

Let us next further enlarge the class of theories, by considering orbifolds of generalized Narain CFTs.

\subsection{Orbifold Actions}

We start with the un-orbifolded theory specified by the lattice $\Lambda$.
Given a group $G$, 
let us consider an action of $g\in G$ on the momentum $p=(p_L, p_R)$:
\begin{align}\label{g_action}
g: (p_L, p_R ) \mapsto   (\theta_L(g) p_L + V_L(g), \theta_R(g) p_R + V_R(g) )  \;.
\end{align}
This is defined by a pair of the rotation matrix $\theta(g)=(\theta_L(g), \theta_R(g))$ and the ``shift vector'' $V(g)=(V_L(g), V_R(g))$,
where the latter is defined modulo the lattice $\Lambda$. In general, there are multiple ways of expressing the same
orbifold action in terms of $\theta$ and $V$. We will see, however, that such ambiguities will not affect the subsequent discussions.
For consistency with the identification of the un-orbifolded theory we need to impose the condition that
$G$ is an automorphism of the lattice $\Lambda$, namely $g \cdot \Lambda = \Lambda$ for any $g\in G$.

The rotations $\theta(g)$ generate the point group,
while the pair $(\theta(g), V(g))$, together with the translations of the original lattice $\Lambda$, generates the space group.
Its elements satisfy the consistency condition
\begin{align}
(\theta(g_1), V(g_1)) \circ (\theta(g_2), V(g_2)) = (\theta(g_1 g_2), V(g_2) +  \theta(g_2) V(g_1)) \;.
\end{align}
The toroidal orbifold is defined by the identification of the momentum under the space group.

For simplicity of presentations, we choose $G=\mathbb{Z}_N$ for the rest of this paper.
We use the symbol $\omega$ to be a generator of $\mathbb{Z}_N$,
and sometimes denote the rotation $\theta(\omega)$ (shift vector $V(\omega)$) simply as $\theta$ ($V$).
In this case, $\omega^n= (\theta, V)^N=(\theta^N, N P_{\omega}V  )  =(1,0) $ with
\begin{align}
P_{\omega}: = (1 + \theta + \theta^2 + \dots + \theta^{N-1} ) /N
\end{align}
being a projector onto $\theta$-invariant states. Let us define 
$I$ to be the invariant lattice of $\Lambda$ under $\omega$ (hence under $G=\mathbb{Z}_N$):
\begin{align}\label{I_def}
I := \left\{ \ell \in \Lambda | \, \omega \cdot \ell = \ell \right\} \;.
\end{align}
This lattice is moduli-independent for symmetric orbifolds $\theta_L=\theta_R$,
and we will in the following restrict our analysis to symmetric orbifolds.\footnote{For asymmetric orbifolds
the invariant lattice $I$ is in itself moduli-dependent, and it becomes more challenging to discuss ensemble averages.}

We define $W$ to be the ``projected shift vector''
\begin{align}
W:=P_{\omega} V \;.
\end{align}
Since $(\theta, V)^N=1$, we need $\theta^N = 1$ and $ N W  \in \Lambda$.
The former condition states that
we can denote the eigenvalues of $\theta_L(\omega)$ and $\theta_R(\omega)$ in terms of ``twist vectors'' $l_j, r_j \in \mathbb{Z}/N\mathbb{Z}$:
\begin{align}
\theta_L(g): \{ e^{2\pi i \, l_j } \} \;, \quad
\theta_R(g): \{ e^{2\pi i \, r_j } \} \;.
\end{align}
The latter condition can be written as a condition on $W$:
\begin{align}\label{P_V}
W \in I(N) / I   \;,
\end{align}
where we defined the $I(N)$ to be
\begin{align}
I(N) := \left\{ N\ell \in \Lambda | \, N \ell \in I \right\} \;.
\end{align}
Note that the ambiguity of the shift of $V$ by an element of $\Lambda$
does not affect the element \eqref{P_V} modulo $I$. 
Moreover, the lattice $I$ and the projected shift vector $W$ is not
affected by the ambiguity of expressing $\omega$ in term of the rotation $\theta$ 
and the shift $V$.

In the following we consider symmetric orbifolds with $\theta_L=\theta_R$ (and hence for example $l_j=r_j$), so that
the orbifold action \eqref{g_action} can be written in a moduli-independent manner:
\begin{align}\label{g_action_sym}
g: p \mapsto   \theta_L(g) p + V(g) \;.
\end{align}
The orbifold projects out the Narain moduli space \eqref{eq:genmod} onto 
a smaller subspace
\begin{align} \label{M_I}
\mathcal{M}_{I}:=
\mathrm{O}_I(\mathbb{Z}) \big\backslash (\rmO(p_I,q_I; \mathbb{R})\big/(\rmO(p_I; \mathbb{R})\times \rmO(q_I; \mathbb{R})) \subset \mathcal{M}_{Q}\;,
\end{align}
where we denoted the signature of $I$ by $(p_I, q_I)$.
In the following we consider an ensemble average over this moduli space.

Notice that in general the lattice $I$ could be trivial---in the examples of $\mathbb{Z}_2$-orbifolds of 
toroidally-compactified theories discussed in \cite{Benjamin:2021wzr}, for example, all the twisted sectors have no moduli dependence,
so that the ensemble average over $\mathcal{M}_I$ is trivial. Nevertheless, since we consider a rather general class of 
orbifolds of general lattices, we will in general have non-trivial invariant lattices and hence the associated moduli space.

\subsection{Torus Partition Functions}\label{subsec:torus}

Let us next discuss the torus partition functions of the theory.
While such partition functions have been extensively studied in the context of heterotic string compactifications,
we are not focusing on specific lattices and will keep the lattice to be a general lattice associated with even quadratic forms.

The torus partition function of the $\mathbb{Z}_N$-orbifolded theory is built out of the 
following $N^2$ partition functions:
\begin{align}
&\textrm{untwisted sector}:  Z(1, \omega^n) \quad (n=0, 1, \dots, N-1) \;, \\
&\textrm{twisted sector}: Z(\omega^m ,\omega^n) \quad (m=1, \dots, N-1; n=0, 1, \dots, N-1)  \;.
\end{align}
Here $Z(g, h)$ denotes the torus partition function 
with twists $g \in G$ (and $h \in G$) along spatial (and temporal) directions:
\begin{align}
Z(g,h) = \begin{picture}(20,20)(0,0)
\put(10,0){\framebox(15,15){}}
\put(15,-7){\mbox{$g$}}
\put(2,4){\mbox{$h$}}
\end{picture} \quad \;.
\end{align}
In general such twists require a compatibility condition $gh=hg$. This condition is automatically satisfied for 
Abelian orbifolds.

The full torus partition function of the $G$-orbifolded theory is given by
\begin{align}
Z_{Q/G} = \frac{1}{N}\left[\overbrace{ \sum_{h \in G} Z(1,h)}^{\textrm{untwisted sector}} +\overbrace{\sum_{\genfrac{}{}{0pt}{}{g, h\in G}{[g,h]=1; g\ne 1} } c(g,h) Z(g,h)}^{\textrm{twisted sector}} \right] \;,
\end{align}
and in the literature the coefficients $c(g,h)$ are chosen appropriately for the modular invariance of the full partition function. We are, however, not imposing the modular invariance in this paper,\footnote{In \cite{Ashwinkumar:2021kav} the modular non-invariance of the torus partition function on the boundary CFT
matches with the framing anomaly of the Chern-Simons theory in the bulk, and the modular non-invariance was a crucial ingredient for the consistency of the discussion.} and for our paper it is enough to discuss individual blocks $Z(g,h)$ separately. 

For the untwisted sectors, the orbifold action is along the temporal direction 
and hence we can still use the same Hilbert space $\mathcal{H}_{Q}$ as in the 
un-orbifolded case, as long as we insert an appropriate operator $U(g)$ representing the twist by $g \in G$:
\begin{align}\label{Z1g}
Z(1,g)=\textrm{Tr}_{\mathcal{H}_{\Lambda}} \left( U(g) \, q^{L_0-\frac{c}{24}} \qbar^{\overline{L}_0-\frac{\cbar}{24}} \right)  \;.
\end{align}

For $g=1$, $Z(1,1)$ is nothing but the un-orbifolded partition function $Z^{\Lambda}$ given in \eqref{ZQ}.
For $Z(1,g)$, from the $g$-action given in \eqref{g_action}, 
we find that the only contribution comes from the subspace of the un-orbifolded Hilbert space 
left invariant under the generator $\omega$; this is described by the 
the $G$-invariant sublattice $I$ \eqref{I_def}.
This means that the partition function \eqref{Z1g} contains the contribution
\begin{align}
\sum_{\ell \in I} e^{i \pi \tau Q_L(\ell)- i \pi \taubar Q_R(\ell)} e^{2\pi i (p_L(\ell) V_L -p_R(\ell) V_R)  } 
= \sum_{\ell \in I} e^{i \pi \tau Q_L(\ell)- i \pi \taubar Q_R(\ell)} e^{2\pi i Q(\ell, W )   }\;,
\end{align}
where in the second line we used $Q(\ell, V)= Q(P_{\omega} \ell, V) = Q( \ell, W )$ for $\ell \in I$.

We also need to take into account contributions from oscillators. 
For $(m l_j, n l_j) = (0,0)$ we have the familiar contribution $1/\eta(\tau)$ from the full set of chiral oscillators.
For $(m l_j, n l_j) \ne (0,0)$ 
the contribution from chiral oscillators is  $\eta(\tau)/ \gentheta{\frac{1}{2} + m l_j}{\frac{1}{2} + n l_j}(\tau)$, 
where  the theta function $\gentheta{\theta}{\phi}(\tau)$ is given by
\begin{align}
\begin{split}
\gentheta{\theta}{\phi}(\tau) &: = \sum_{n\in \mathbb{Z}} q^{\frac{1}{2} (n+\theta)^2 } e^{2\pi i(n+\theta) \phi } \\
&= \eta(\tau) \, e^{2 \pi i \theta \phi} q^{\frac{1}{2} \theta^2-\frac{1}{24}} \prod_{k=1}^{\infty}
(1+q^{k+\theta-\frac{1}{2} } e^{2\pi i \phi}) (1+q^{k-\theta-\frac{1}{2} } e^{-2\pi i \phi})  \;.
\end{split}
\end{align}
There are similar oscillator contributions for anti-chiral oscillators, with $\tau$ replaced by $\taubar$.

By collecting all the ingredients, we obtain
\begin{align}
Z(\omega^m,\omega^n) =
\begin{cases}
Z_\Lambda  =\displaystyle\frac{ \vartheta_{\Lambda}}{\eta^p {\overline{\eta}^q}} & (m,n)= (0,0) \;,\\
 Z^{I^{\perp}} (\omega^m,\omega^n) Z^{I} (\omega^m,\omega^n) 
& (m,n)\ne (0,0) \;,
\end{cases} 
\end{align}
The building blocks, $Z_{I^{\perp}}$ for the oscillators and $Z_{I}$ for the lattice, are given by 
\begin{align}
&Z^{I^{\perp}} (\omega^m,\omega^n)  =  C(I, \vartheta) \prod_{j:\, l_j\ne 0} \frac{\eta(\tau)}{ \gentheta{\frac{1}{2} + m l_j}{\frac{1}{2}+ n l_j}(\tau)}
 \prod_{j:\, r_j\ne 0}  \frac{\etabar(\taubar)}{ \gentheta{\frac{1}{2} + m r_j}{\frac{1}{2}+ n r_j}(\taubar)}  \;,  \label{Z_I_perp}\\
&Z^I (\omega^m,\omega^n) = \frac{\vartheta^I_{(m W , n W )}(\tau, \taubar;m) }{\eta(\tau)^{p_I} \etabar(\taubar)^{q_I}} \;,
\label{Z_I}
\end{align}
and $C(I, \vartheta)$ is a moduli-independent numerical factor which is not important for our analysis.\footnote{This includes
an integer $\sqrt{\det\!{}^{\prime}(1-\theta)/|I^*/I|}$ counts the number of fixed points under the $\theta$-rotation \cite{Narain:1986qm}. Here $I^{*}$ is defined in \eqref{I_dual_def} and 
the primed determinant $\det\!{}^{\prime}$ is defined by disregarding the zero eigenvalues.}

In \eqref{Z_I} we defined a generalized version of the theta function associated with the invariant lattice $I$: 
\begin{align} \label{def_Theta_I}
\vartheta^I_{(\delta, \eta)} (\tau, \taubar;m):= \sum_{\ell \in I+\delta } e^{i \pi \tau Q_L(\ell)- i \pi \taubar Q_R(\ell)} e^{2i\pi  Q_I(\ell, \eta)}\;,
\quad (\delta \in \mathbb{R}^{p_I+q_I}/ I, \quad \eta \in \mathbb{R}^{p_I+q_I})\;, 
\end{align}
where we have shown the dependence on the moduli $m\in \mathcal{M}_I$,
and $Q_I$ denotes the quadratic form associated with the lattice $I$.
Since $Q_I$ is a restriction of $Q$, for notational simplify we mostly use the same symbol $Q$ to denote the quadratic form $Q_I$;
however, the distinction matters when we discuss the determinants of the quadratic forms, for example.

We can verify from the definition that $\vartheta_{I, (\delta, \eta)}$ is indeed 
preserved when we shift $\delta$ by an element of $I$, as suggested in \eqref{def_Theta_I}.
By contrast, when we shift $\eta$ by an element of $\beta\in I^*$ we obtain an extra phase factor
\begin{align} 
\vartheta^I_{(\delta, \eta+ \beta)} (\tau, \taubar;m)=e^{2i\pi  Q(\delta, \beta)} \,  \vartheta^I_{(\delta, \eta)} (\tau, \taubar;m) 
\;,
\end{align}
and the phase factor is in general non-trivial even when $\beta \in I$.

By relabeling the sum as $\ell \to -\ell$ we obtain
\begin{align}
\vartheta^I_{(\delta, \eta)} (\tau, \taubar;m)= \vartheta^I_{(-\delta, -\eta)} (\tau, \taubar;m) \;.
\end{align}

Let us denote the signature of $Q_I$ by $(p_I, q_I)$, and 
define the dual lattice $I^{*}$ of $I$ as a dual inside $\mathbb{R}^{p_I+q_I}$ (and not $\mathbb{R}^{p+q}$):
\begin{align}\label{I_dual_def}
I^*:= \left\{ x \in \mathbb{R}^{p_I+q_I} \,\Big|\, Q(x, \ell) \in  \mathbb{Z} \quad  (\forall \ \ell \in I) \right\} \;.
\end{align}
Similar to the case of $\Lambda$, we define the discriminant for $I$ to be
\begin{align}
\mathscr{D}_I = I^*/I \;.
\end{align}

For $T$ and $S$ generators of $\SL(2, \mathbb{Z})$ we have the modular transformation (see \Cref{appendix:theta} for derivation):\footnote{
Special cases of these transformation rules appeared in string theoretic literature \cite[section 3]{Senda:1987pf}.}
\begin{align}
\label{Theta_T}
  &T: \quad \vartheta^I_{(\delta, \eta)}(\tau+1; m) =   e^{-\pi i Q(\delta)}\,  \vartheta^I_{(\delta, \eta+\delta)}(\tau;m) \;, \\
\label{Theta_S}
  &S: \quad \vartheta^I_{(\delta, \eta)}\left(-\frac{1}{\tau}; m\right) = \frac{e^{-\frac{ i \pi (p_I-q_I)}{4}} }{\sqrt{\textrm{det}\, Q_I}} \tau^{\frac{p_I}{2}}   \taubar^{\frac{q_I}{2}} e^{2\pi i Q(\delta, \eta)}  \sum_{\gamma\in \mathscr{D}_I} \vartheta^I_{  (\gamma+\eta, -\delta)}(\tau; m) \;.
\end{align}
In terms of the ratio
\begin{align}
Z^I_{(\delta, \eta)}(\tau;m) := 
\frac{\vartheta^I_{(\delta, \eta)}(\tau; m) }{\eta^{p_I}(\tau) \etabar^{q_I}(\taubar)} \;,
\end{align}
we have
\begin{align}
&T: Z^I_{(\delta, \eta)}(\tau+1; m) = e^{-\frac{2\pi i (p_I-q_I)}{24}}e^{-\pi i Q(\delta)} \,Z^I_{(\delta, \eta+\delta)}(\tau; m) \;,\\
&S: Z^I_{(\delta, \eta)}\left(-\frac{1}{\tau}; m \right) = \frac{1}{\sqrt{\textrm{det}\, Q_I}} \, e^{2\pi i Q(\delta, \eta)}  \sum_{\gamma\in \mathscr{D}_I} Z^I_{(\gamma +\eta, -\delta)}(\tau; m) \;.
\end{align}

For our practical purposes, we do not need fully general choices of $\delta, \eta$,
and we only need those needed for the discussion of the modular transformations of
$Z^I_{(mW, nW)}$ \eqref{Z_I}.
This picks up the choice $\delta= \alpha+mW$ and $\eta=\beta + nW$
for $\alpha, \beta \in \mathscr{D}_I$ and $m, n \in \mathbb{Z}_N$.
The only effect of $\beta$, however, is to change the overall normalization of the partition function 
\begin{align}\label{delta_eta_shift} 
\vartheta^I_{(\alpha + mW,  \beta + nW)} (\tau; m)=e^{2\pi i Q(\alpha +mW, \beta)}\, \vartheta^I_{(\alpha + mW, nW)} (\tau; m) \;,
\end{align} 
and hence we can set $\beta=0$ when we are discussing basic building blocks up to overall normalization factors.
The modular transformations of the blocks $Z^I_{(\alpha+mW , nW )}$ are given by, after using \eqref{delta_eta_shift}: 
\begin{align}\label{TS_I}
&T: Z^I_{(\alpha+mW , nW )}(\tau+1; m) 
= e^{-\frac{2\pi i (p_I-q_I)}{24}} 
e^{-\pi i m^2 Q(W)  + \pi i Q(\alpha)}  
\, Z^I_{(\alpha+mW , (m+n)W)}(\tau;m) \;,\\
&S: Z^I_{(\alpha+mW , nW )}\left(-\frac{1}{\tau} ;m\right) 
= \frac{1}{\sqrt{\textrm{det}\, Q_I }} \, e^{2 \pi i mn Q( W)}  \sum_{\gamma\in \mathscr{D}_I}
e^{-2\pi i Q(\gamma, \alpha)}
Z^I_{ (\gamma+nW, -mW)}(\tau; m) \;.
\end{align}
As we prove in \Cref{appendix:congruence},
the partition functions $Z_{I, (\alpha+mW , nW )}$ are modular functions with respect to a 
congruence subgroup $\Gamma(N^2 L)$ of $\PSL(2, \mathbb{Z})$,
where an integer $L$ is a integer multiple of the level $L_I$ of the quadratic form $Q_I$, where the level is defined as the smallest integer such that $L_I Q_I^{-1}$ is integral.\footnote{The level $L_I$ was denoted as $N$ in \cite{Ashwinkumar:2021kav}.}

By comparing the modular transformations \eqref{TS_I} with those of the un-orbifolded case \eqref{CFT_Z_TS},
one finds that the $(m,n)\ne 0$ sector we have the same transformation as the un-orbifolded case (represented by $m=n=0$), the only difference being that 
\begin{enumerate}
    \item the lattice $\Lambda$ replaced by the $G$-invariant lattice $I$ and 
    \item the discriminant $\mathscr{D}=\Lambda^*/\Lambda$ is replaced by 
          \begin{align}\label{hat_D}
          \hat{\mathscr{D}}:=\mathscr{D}_I \oplus \mathscr{D}_N \;.
           \end{align}
          which is larger than $\mathscr{D}_I=I^*/I$ by a factor  of
             $\mathscr{D}_N:= \mathbb{Z}_N$,
          the choice of which is determined by the projected shift vector $W$.
 \end{enumerate}         
The braiding between the two elements of $\mathscr{D}_N$  is determined by 
the quadratic form $Q$ and the projected shift vector $W$:
\begin{align}\label{D_I}
B(m,n)=e^{2\pi i Q(mW, nW)}  \;, \quad (m, n \in \mathbb{Z}_N ) \;.
\end{align}

For the special case $N=2$ we have 
\begin{align} 
\vartheta^I_{(\alpha+mW ,  nW )} (\tau, \taubar):= \sum_{\ell \in I+\alpha + mW} e^{i \pi \tau Q_L(\ell)- i \pi \taubar Q_R(\ell)} (-1)^{nQ(\ell, 2W) }\;,
\end{align}
with $\alpha, \beta \in I^*$, $W\in I/2$ and $m,n=0, 1$. 
This expression is reminiscent of the theta function for spin CFTs in our previous paper \cite{Ashwinkumar:2021kav},
and the projected shift vector $2W\in I$ plays a role analogous to the Wu class \cite{MR0440554,MR0145525} for the spin CFT.

The modular transformation for the orbifolded lattice theta function \eqref{def_Theta_I} is given by 
(see \Cref{appendix:theta} for derivation)
 \begin{align}\label{theta_SL2}
     \vartheta^I_{(\delta , \eta )}(\tau_{M};m)
     & =
     \mu_{(\delta, \eta)\cdot M}
     \sum_{\gamma \in \mathscr{D}_I}\mathcal{U}_{\alpha,\gamma}^I(M) \,\vartheta^I_{ (\gamma, 0)+(\delta, \eta)\cdot M} (\tau;m) \;,
 \end{align}
Here we defined a phase $\mu_{(\delta, \eta)}$ by
\begin{align} \label{mu_def}
\mu_{(\delta, \eta)} := e^{-i\pi Q(\delta,\eta) }
\end{align}
so that $\mu_{(-\delta, -\eta)} = \mu_{(\delta, \eta)}$ and $\mu_{(\delta, \eta) \cdot M} =  \mu_{(a\delta + c\eta, b\delta + d\eta)}=  e^{-i\pi Q(a\delta + c\eta, b\delta + d\eta)}$. 
The matrix $\mathcal{U}^I_{\alpha,\gamma}(M, \tau)$ (for $M$ given as in \eqref{M_def} with entries $a, b, c, d$) is given by the essentially same expression as 
 $\mathcal{U}^{\Lambda}_{\alpha, \gamma}(M, \tau)$, with $\Lambda$ replaced by $I$: 
 \begin{align} \label{Ucal_def_I}
     \mathcal{U}^{I}_{\alpha, \gamma}(M, \tau) := 
     \begin{cases}
         \displaystyle   (c\tau+d)^{\frac{p}{2}} (c\taubar+d)^{\frac{q}{2}}  \lambda^{I}_{\alpha, \gamma}(M)  &  (c\ne 0)\;, \\
         \displaystyle   \mathcal{U}^{I}_{\alpha, \gamma}(T, \tau) := e^{i\pi  Q(\alpha  )} \delta_{\alpha, \gamma} & (c=0) \;,
     \end{cases}
\end{align}
and
 \begin{align} 
 \lambda^I_{\alpha, \gamma}(M, \tau) := \frac{1}{\sqrt{\textrm{det}\, Q_I} }\,  e^{-\frac{\pi i}{4}(p_I-q_I)} c^{-\frac{p_I+q_I}{2}}   \sum_{\ell_c \in I/( c  I)} e^{\frac{\pi i}{ c  }\left(a Q[\ell_c+\alpha ] -2Q[\ell_c +\alpha , \gamma]+ d   Q[\gamma] \right)} \;.
\end{align}

As in the un-orbifolded case \eqref{M_consistent}, we have
 the consistency relations
\begin{align}\label{M_consistent_I}
     \mathcal{U}^{I}_{\alpha, \beta}(M \cdot M', \tau)  =    \sum_{\gamma \in \mathscr{D}_{\Lambda}}  \mathcal{U}^{I}_{\alpha, \gamma}(M, \tau_{M'})  \,\mathcal{U}^{I}_{\gamma, \beta}(M', \tau)  \;,
\end{align}
for two $\SL(2, \mathbb{Z})$ matrices $M, M'$.
 We also have the sign-flip relations
 \begin{align} 
 \mathcal{U}^I_{\alpha, \gamma}(M, \tau) =  \mathcal{U}^I_{\alpha, -\gamma}(-M, \tau) \;, 
\end{align}
which ensures that the right hand side of \eqref{theta_SL2} does not depend on the choice of the $\SL(2, \mathbb{Z})$ representative $M$.
 
For the special case $\delta= \alpha+mW, \eta = \beta+nW$, we have
\begin{align}\label{theta_SL2_special}
     \vartheta^I_{(\alpha + mW , \beta+ nW )}(\tau_{M};m)
     & =
     \mu_{(\alpha+ m W , \beta+n W)\cdot M}
     \sum_{\gamma \in \mathscr{D}_I}\mathcal{U}_{\alpha,\gamma}^I(M) \,\vartheta^I_{(\gamma,0 ) + (\alpha+mM, \beta+nW)\cdot M}(\tau;m) \;,
 \end{align}

\section{Ensemble Average}\label{sec:ensemble}
\subsection{Unorbifolded Case}

The ensemble average of the theta function $\vartheta_{Q,\alpha}$ over the moduli space \eqref{eq:genTDa}\footnote{We impose $p+q>4$  for the convergence of the integral.}
\begin{align}\label{Siegel_Qh}
\langle \vartheta^Q_{\alpha}\rangle_{\mathcal{M}_Q} (\tau) &:= 
\frac{1}{\textrm{Vol}(\widetilde{\mathcal{M}}_Q ) } \int_{\tilde{\mathcal{M}}_Q} [dm] \, \vartheta^Q_{\alpha}(\tau;m) 
\end{align}
 is given by the Siegel-Eisenstein series \cite{MR67930,Siegel_Lecture,Maloney:2020nni,Ashwinkumar:2021kav} (henceforth referred to simply as the Eisenstein series)
 \begin{align}
\langle \vartheta^{Q}_{\alpha}(\tau) \rangle_{\mathcal{M}_\Lambda}=E^{Q}_{\alpha}(\tau) := \delta_{\alpha \in \Lambda} +
 \sum_{(c,d)=1,\, c>0} \, \frac{\gamma^\Lambda_{\alpha}(c,d)}{ (c\tau+d)^{\frac{p}{2}} (c\taubar+d)^{\frac{q}{2}}}  \;,
 \label{eq:E_as_sum}
\end{align}
associated with the lattice $\Lambda$, where $\delta_{\alpha}=1$ for $\alpha\in \Lambda$, and $\delta_\alpha=0$ for $\alpha\notin \Lambda$.
The factor $\gamma^\Lambda_{\alpha}(c,d)$ is given by 
\begin{align}\label{gamma_def}
\gamma^{\Lambda}_{\alpha}(c,d):=\lambda_{\alpha, 0}(M^{-1}) = \frac{1}{\sqrt{\textrm{det}\, Q} } \, e^{-\frac{\pi i (p-q)}{4}} (-c)^{-\frac{p+q}{2}} \sum_{\ell_c  \in \Lambda/c\Lambda} \exp\left[ -\pi i \frac{d}{c} Q(\ell_c +\alpha) \right] \;.
\end{align}
with $M$ given as in \eqref{M_def}.  Note that this depends only on the two entries $c,d$ of the matrix $M$, as opposed to general $\lambda_{\alpha, \beta}(M)$ given in \eqref{lambda_def}.

Since the Eisenstein series is an ensemble average of the theta function, and since the modular transformation matrices $\mathcal{U}$ in \eqref{theta_Lambda_SL2}
are independent of the moduli, the modular transformation of the Eisenstein series should be the same as in \eqref{theta_Lambda_SL2}:
\begin{align}\label{E_Lambda_SL2}
     E^{\Lambda}_{\alpha}(\tau_{M})
     & = \sum_{\beta \in \mathscr{D}_\Lambda}\mathcal{U}^{\Lambda}_{\alpha, \beta}(M, \tau) \,E^{\Lambda}_{\beta}(\tau) \;.
 \end{align}
 We can verify this expression explicitly from the definition \eqref{eq:E_as_sum} with the help of \eqref{gamma_def} and \eqref{M_consistent}.
 Finally, we have the counterparts of \cref{theta_flip} and \cref{eq:genthetaTD}:
 \begin{align}     
      E^{\Lambda}_{\alpha}(\tau) = E^{\Lambda}_{-\alpha}(\tau) \;, \quad
      E^{\Lambda}_{g\cdot \alpha}(\tau) =E^{\Lambda}_{\alpha}(\tau) \;.
 \end{align}

\subsection{Orbifolded Case}

Let us next discuss the ensemble average of the partition functions over the projected Narain moduli space $\mathcal{M}_I$ \eqref{M_I}.\footnote{We impose $p+q>4$  for the convergence of the integral.} We evaluate the ensemble average of the 
theta function $\vartheta_{I, (\delta, \eta)}$ introduced in \eqref{def_Theta_I},
which we denote as $E^{I}_{(\delta, \eta)}$ (in anticipation of the fact this is a certain generalization of the Siegel-Eisenstein series) :
\begin{align}
\langle \vartheta_{I, (\delta, \eta)} \rangle_{\mathcal{M}_I}(\tau, \taubar) = 
\frac{1}{\textrm{Vol}(\mathcal{M}_I ) } \int_{\mathcal{M}_{I}}[dm] \, \vartheta_{I, (\delta, \eta)}(\tau, \taubar;m) \;.
\end{align}

To evaluate this average we need some generalization of the Siegel-Weil formula,
which does not seem to exist in the literature.\footnote{The exception is Ref.~\cite{Dong:2021wot}, where the authors considered 
the special cases of CFTs defined by even self-dual lattices and vanishing chiral central charge ($p=q$ in our notation).
They considered the $\SU(N)$ WZW models at level $k$,
which describes the $\mathbb{Z}_N$-quotient of a product theory, i.e., Narain theory times the parafermion theory.
The theta functions in the two references are related as 
$\vartheta^{\rm (here)}_{(\delta, \eta)}(\tau) = e^{\pi i \delta\cdot \eta}\vartheta^{\rm (there)}_{(-
\delta,\eta)}(\tau)$.  Since the difference is only an overall constant, the same trivially propagates into the 
definition of the Eisenstein series after the ensemble averages. It should be noticed, however,
that the phase factors transform non-trivially under the modular transformation 
and hence we obtain slightly different expressions for the Poincar\'e sum.}

Fortunately for us, the Siegel-Weil formula was re-derived in our previous paper \cite[section 2.3]{Ashwinkumar:2021kav}
and we can use the same logic to derive the formula necessary for our present purposes,
as explained below.

We find that the ensemble average $\langle \vartheta_{I, (\delta, \eta)} \rangle_{\mathcal{M}_I}$ evaluates to
\begin{align} \label{E_sum_1}
    E^I_{(\delta, \eta)} (\tau, \taubar) &:= \delta_{\delta\in I} +   \frac{1}{\sqrt{\textrm{det}\, Q_I} }\, e^{-\frac{\pi i}{4}(p_I-q_I)}
    \sum_{(c,d)=1,c>0} c^{-\frac{p_I+q_I}{2}}(c\tau+d)^{-\frac{p_I}{2}}(c\bar{\tau}+d)^{-\frac{q_I}{2}}  \nonumber\\
&\quad \times   \mu_{(\delta, \eta)\cdot M^{-1}}
 \lambda^I_{0, -d \delta + c \eta}(M^{-1}) \;,
\end{align}
where $M$ is a $\PSL(2, \mathbb{Z})$ matrix with the second row given by $(c,d)$ as in \eqref{M_def},
and we used previously-defined $\mu_{(\delta, \eta)}$
 as in \cref{mu_def}, as well as 
 \begin{align}
     \lambda^I_{0, -d \delta + c \eta}(M^{-1}) &= \sum_{\ell_c \in I/( c  I)} e^{-\frac{\pi i}{ c  }\left(d Q[\ell_c ] -2Q[\ell_c , -d \delta + c \eta]+ a   Q[-d \delta + c \eta] \right)} 
\end{align}
as in \cref{lambda_def} (previously defined for $\Lambda$).
This is one of the main technical results of this paper. 
We can slightly simplify the formula into
\begin{align}\label{E_sum_2}
    E^I_{(\delta, \eta)} (\tau, \taubar) 
   &= \delta_{\delta\in I} +   \frac{1}{\sqrt{\textrm{det}\, Q_I} }\, e^{-\frac{\pi i}{4}(p_I-q_I)}
    \sum_{(c,d)=1,c>0} c^{-\frac{p_I+q_I}{2}}(c\tau+d)^{-\frac{p_I}{2}}(c\bar{\tau}+d)^{-\frac{q_I}{2}}  \nonumber\\
&\quad \times  \sum_{\ell_c \in I/( c  I) + \delta} e^{-\frac{\pi i}{ c  }\left(d Q[\ell_c ] -2c Q[\ell_c  , \eta]+  c Q[\delta, \eta] \right)}  \;.
\end{align}
In the latter form it is manifest that the expression reduces to the un-orbifolded case \eqref{eq:E_as_sum} for $N=1$ (and hence $\delta=\eta=0$).

We call the function in \eqref{E_sum_1} and \eqref{E_sum_2} the ``orbifold Eisenstein series,''
since this is an orbifold analogue of the non-holomorphic Eisenstein series.
The orbifolded Eisenstein series is a modular form with respect to a congruence subgroup $\Gamma(N^2 L_I)$ with $L_I$
being directly proportional to the level of $Q_I$; see also  \Cref{appendix:congruence}. As far as we are aware, there is no literature discussing this Eisenstein-like series
in the mathematical literature (apart from special cases). It would be interesting to further study the properties of this function.
 
For application to orbifolds, we set $(\delta,\eta) = (\alpha + mW, \beta + nW)$
with $\alpha, \beta \in \mathscr{D}_I$ and $m, n \in \mathbb{Z}_N$,
leading to
\begin{align}\label{E_I}
    E^I_{(\alpha+mW, \beta+nW)} (\tau) =\delta_{\delta\in I} +  
    \sum_{(c,d)=1,c>0}  \gamma^I_{(\alpha, \beta; m,n )}(c,d) (c\tau+d)^{-\frac{p_I}{2}}(c\bar{\tau}+d)^{-\frac{q_I}{2}} \;,
\end{align}
with
\begin{align}
\gamma^I_{(\alpha, \beta; m,n )}(c,d) &=\frac{1}{\sqrt{\textrm{det}\, Q_I} }\, e^{-\frac{\pi i}{4}(p_I-q_I)}
    c^{-\frac{p_I+q_I}{2}} \mu_{(\alpha+mW,  \beta+ nW) \cdot M^{-1}} \lambda^I_{0, -d (\alpha+ mW) + c (\beta +nW)}(M^{-1}) \\
& = \frac{1}{\sqrt{\textrm{det}\, Q_I} }\, e^{-\frac{\pi i}{4}(p_I-q_I)}
    c^{-\frac{p_I+q_I}{2}} \nonumber\\
    & \qquad \times\sum_{\ell_c \in I/( c  I) + \alpha + mW} e^{-\frac{\pi i}{ c  }\left(d Q[\ell_c  ] -2c Q[\ell_c  , \beta+nW]+  c Q[\alpha+ mW, \beta+nW] \right)}  \;.
\end{align}
The averaged partition functions are then
\begin{align}\label{EI_average} 
\langle Z^I_{(\alpha+mW,\beta+nW)}(\tau, \taubar) \rangle
&=
\frac{E^I_{(\alpha+mW,\beta+nW)}(\tau,\taubar)}{\eta(\tau)^{p_I}\overline{\eta}(\overline{\tau})^{q_I}} \nonumber\\
&=\frac{\delta_{mW\in I}}{\eta(\tau)^{p_I}\overline{\eta}(\overline{\tau})^{q_I}}
+\sum_{g\in \Gamma_\infty\backslash\PSL(2,\mathbb{Z}), g\ne 1}
\epsilon(g)^{p_I - q_I}
\frac{ \gamma^I_{(\alpha, \beta; m,n)}(g)}{\eta(g\cdot \tau)^{p_I}\overline\eta(g\cdot \overline{\tau})^{q_I}}\;,
\end{align}
where the phase $\epsilon(g)$ 
 is the multiplier system for the Dedekind eta function and is described in \Cref{appendix:modular_transformation}.
Their modular transformations are given by
\begin{align}\label{TS_I_averaged}
&T: \langle Z^I_{(\alpha+mW , nW )}(\tau+1; m) \rangle = e^{-\frac{2\pi i (p_I-q_I)}{24}} 
 e^{-\pi i Q(m W) }  e^{\pi i Q(\alpha) }
\, \langle Z^I_{(\alpha+mW , (m+n)W)}(\tau;m) \rangle \;,\\
&S: \langle Z^I_{(\alpha+mW , nW )}\left(-\frac{1}{\tau} ;m\right) \rangle = \frac{1}{\sqrt{\textrm{det}\, Q_I }} \, e^{2 \pi i Q(mW, n W)}  \sum_{\gamma\in \mathscr{D}_I}e^{-2\pi i Q(\alpha, \gamma) }  \langle Z^I_{ (\gamma-\alpha-mW, nW)}(\tau; m) \rangle \;.
\end{align}

\paragraph{Modular Transformation}

Since the orbifold Eisenstein series is defined by the ensemble average of the orbifold theta function,
we expect the two have the same modular transformation properties
\begin{align}\label{E_SL2}
     E^I_{(\delta,\eta)}(\tau_{M'})
     & = \mu_{(\delta, \eta)\cdot M'} \sum_{\gamma \in \mathscr{D}_I}\mathcal{U}_{0,\gamma}(M') \,E^I_{(\gamma,0) +(\delta, \eta) \cdot M'   }(\tau) \;,
 \end{align}
 and in particular
 \begin{align}\label{E_TS}
&T: \quad E^I_{(\delta, \eta)}(\tau+1) =   e^{-\pi i Q(\delta, \eta)}\,  E^I_{(\delta, \eta+\delta)}(\tau) \;, \\
&S: \quad E^I_{(\delta, \eta)}\left(-\frac{1}{\tau}\right) = \frac{e^{-\frac{ i \pi (p_I-q_I)}{4}} }{\sqrt{\textrm{det}\, Q_I}} \tau^{\frac{p_I}{2}}   \taubar^{\frac{q_I}{2}} e^{-\pi i Q(\delta, \eta)}  \sum_{\gamma\in \mathscr{D}_I} E^I_{  (\gamma-\eta, \delta)}(\tau) \;.
\end{align}

 To verify this, let us write $M'=\left( \begin{array}{cc} a' & b' \\ c' & d' \end{array} \right)$
 and
 \begin{align}
    E_{I,(\delta, \eta)} (\tau) 
&= \delta_{\delta\in I} +     \sum_{(c,d)=1,c>0} \left(\frac{-c\tau+a}{c\tau+d} \right)^{-\frac{p}{2}} \left(\frac{-c\taubar+a}{c\bar{\tau}+d} \right)^{-\frac{q}{2}}  \nonumber\\
&\qquad \times  \mu_{(\delta, \eta) \cdot M^{-1}}
\sum_{\gamma\in \mathscr{D}_I}  \delta_{\gamma +d\delta -c\eta \in I} \,\mathcal{U}_{0, \gamma}(M_{c,d}^{-1}, \tau) \;. 
\end{align}
When evaluating $E^I_{(\delta,\eta)}(\tau_{M'})$
we can use from \eqref{M_consistent_I}
\begin{align}
 \mathcal{U}^I_{0, \gamma}(M_{c,d}^{-1}, \tau_{M'}) = \sum_{\beta \in \mathscr{D}_I} \mathcal{U}^U_{0, \beta}\left((M'^{-1}  M_{c,d})^{-1}, \tau \right) \, \mathcal{U}^I_{\beta, \gamma}(M'^{-1}, \tau_{M'})
 \end{align}
 and
 \begin{align}
 \mathcal{U}_{\beta, \gamma}(M'^{-1}, \tau) =
\mathcal{U}_{\beta, \gamma}(M'^{-1}, \tau_{M'}) \left(\frac{-c\tau+a}{c\tau+d} \right)^{-\frac{p}{2}} \left(\frac{-c\taubar+a}{c\bar{\tau}+d} \right)^{-\frac{q}{2}} \;.
 \end{align}
In the sum over $M$ we can replace $M$ by $M'M$, to obtain
 \begin{align}
    E_{I,(\delta, \eta)} (\tau_{M'}) 
&= \delta_{\delta\in I} +    \sum_{\beta \in \mathscr{D}_I}    \mathcal{U}_{\beta, \gamma}(M'^{-1}, \tau_{M'})
\sum_{(c,d)=1,c>0}  \mu_{(\delta, \eta) \cdot M' \cdot M^{-1}}
\sum_{\gamma\in \mathscr{D}_I}  \delta_{\gamma +d\delta -c\eta \in I} \, \mathcal{U}_{0, \beta}(M_{c,d}^{-1}, \tau)  \;. 
\end{align}

\paragraph{Proof}

We can now prove \eqref{EI_average}.
Let us consider the function $F_Q:=E_Q  -\langle \vartheta_Q \rangle$.
This function satisfies the following three properties:
\begin{enumerate}
\item $F_Q$ is a modular form for a particular congruence subgroup, denoted $\Gamma$, of $\mathrm{PSL}(2,\mathbb{Z})$; 
see \Cref{appendix:congruence}. This congruence subgroup depends on the dimension of the lattice/quadratic form, specifically if whether the lattice dimension is even or odd.

\item $F_Q$ has no singularities at the cusps of $\mathbb{H}/\Gamma$, 
and hence is square-integrable. This is because $E_Q$ and $\vartheta_Q$ have the same behaviour at the cusps; see \Cref{cuspasymp}.

\item $\tau_2^{(p_I+q_I) / 4}F_Q$ is an eigenfunction of the weight $k$ Laplacian
\ie 
\square_k:=-\tau_2^2\left(\partial_1^2+\partial_2^2\right)+i k \tau_2 \partial_1,
\fe
where $k=(p_I-q_I)/2$, with eigenvalue $-\frac{((p_I+q_I) / 4-1)(p_I+q_I)}{4}$. This follows since
  \begin{align}
    \left(\tau_2(\partial_1^2+\partial_2^2)+\frac{p_I+q_I}{2}\partial_2+\frac{i(q_I-p_I)}{2}\partial_1\right) F_Q(\tau)=0 \;.
  \end{align}
\end{enumerate}

These facts are enough to conclude $F_Q=0$ as in \cite[section 2.3]{Ashwinkumar:2021kav}. 
We have thus proven the identity \eqref{EI_average}. 

\paragraph{Poincar\'e Sum}

We can interpret the right hand side of \cref{EI_average} as a sum over 
geometries in the bulk. Using the transformation law of the Dedekind eta function from \Cref{appendix:modular_transformation}, we have
\begin{align}\label{poincare_1} 
    \langle Z^I_{(\alpha+mW,\beta+nW)}(\tau, \taubar) \rangle
&= \sum_{g\in \Gamma_\infty\backslash\PSL(2,\mathbb{Z})} 
\epsilon(g)^{\sigma_I}\frac{ \gamma^I_{(\alpha, \beta; m,n)}(g)}{\eta(g\cdot \tau)^{p_I}\overline\eta(g\cdot \overline{\tau})^{q_I}}
\;.
\end{align}
This suggests that the average has a holographic interpretation.
However, \cref{poincare_1} describes only one sector of the boundary CFT, and we must cast the entire partition function as a Poincar\`{e} sum 
to justify a holographic interpretation. This can be done by folding in the twisted sector partition functions $Z^{I^{\perp}}(\omega^m,\omega^n)$ into the Poincar\`{e} sum in the same way as the Dedekind eta functions above. Thus we write the averaged twisted sector contribution to the partition function as
\begin{align}\label{poincare_2}
\begin{split}
    \langle Z^{I^\perp}(\omega^m,\omega^n)&Z^I(\omega^m,\omega^n)\rangle =\\ &C(I, \vartheta)\sum_{g\in \Gamma_\infty\backslash\PSL(2,\mathbb{Z})}\left(\prod_{j:\, l_j\ne 0} \frac{\eta(g\cdot\tau)}{ \gentheta{\frac{1}{2} + m^\prime l_j}{\frac{1}{2}+ n^\prime l_j}(g\cdot\tau)}
 \prod_{j:\, r_j\ne 0}  \frac{\etabar(g\cdot\taubar)}{ \gentheta{\frac{1}{2} + m^\prime r_j}{\frac{1}{2}+ n^\prime r_j}(g\cdot \taubar)}\right)\\
 &\hspace{4cm} \times\frac{ \Psi_{(\alpha, \beta; m,n)}(g)}{\eta(g\cdot \tau)^{p_I}\overline\eta(g\cdot \overline{\tau})^{q_I}}
 \end{split}
 \;,
\end{align}
where
 \begin{equation}
     \begin{pmatrix}
         m^\prime \\ n^\prime
     \end{pmatrix} = \begin{pmatrix}
         d & -c\\
         -b & a
     \end{pmatrix}\begin{pmatrix}
         m\\ n
     \end{pmatrix}
 \end{equation}
 and the phase $\Psi(g)$ is a combination of phases from the twist and shift sectors:
\begin{align}
     \Psi_{(\alpha, \beta; m,n)}(g) 
     &= \prod_{j:\, l_j\ne 0}\prod_{j:\, r_j\ne 0}  \frac{\myS(g,mr_j,nr_j)}{\myS(g,ml_j,nl_j)} \epsilon(g)^{\sigma_I} \gamma^I_{(\alpha, \beta; m,n)}(g)
     \;.
\end{align} 
The phases $\myS(g,\alpha,\beta)$ are defined as 
\begin{equation}
\label{myS_def}
\myS(M,\alpha,\beta) := e^{i\pi (\alpha(\beta+1)-\alpha^\prime(\beta^\prime+1))} \;.
\end{equation}
(see \Cref{appendix:modular_transformation} for further details).
The untwisted sector average can also be written as a Poincar\`{e} sum~\cite{Ashwinkumar:2021kav}, and so we have shown that the entire averaged CFT partition function can be cast as a Poincar\`{e} sum. Note that the summand in \cref{poincare_2} appears to explicitly depend on the upper two components of the $\SL(2,\mathbb{Z})$ matrix $g$. However, this is illusory since this dependence in the theta functions cancels out with a similar dependence in the phase $\Psi(g)$. 

\subsection{Flavored Case}

Let us next go back to the un-orbifolded case $N=1$.

In \cite{Datta:2021ftn},  ensemble averaging for Jacobi theta functions defined with respect to even self-dual lattices was considered. Here, the average was taken over $\rmO(D,D; \mathbb{Z})\backslash \rmO(D,D;\mathbb{R})$ instead of the usual Narain moduli space. 

We shall generalize this analysis to the case of Jacobi theta functions associated with general even integer lattices which take the form
\ie \label{jactoav}
\vartheta_{\alpha}(\tau, z)=\sum_{\ell} \exp \left(i \pi\left\{\tau Q_L(\ell+\alpha)-\bar{\tau} Q_R(\ell+\alpha)\right\}\right) e^{2 \pi i [Q_L(z,\ell +\alpha )- Q_R(z,\ell +\alpha )]} \;,
\fe
where the integration is now over $\rmO_Q(p,q;\mathbb{Z}) \backslash \rmO(p,q; \mathbb{R})$.\footnote{As in previous cases, we impose $p+q>4$  for the convergence of the integral.}  These functions have a well-defined transformation under $\textrm{PSL}(2,\mathbb{Z})$, as shown in \Cref{appendix:theta}.

Let us identify $\rmO(p,q;\mathbb{R})$ as a set of linear transformations which preserve the quadratic form $\eta_{AB}Z^AZ^B$, where $A,B= 1, 2, \ldots, p+q$, and where $\eta_{A B}=\operatorname{diag}\left(1^p,-1^q\right)$. The generators of $\rmO(p,q,\mathbb{R})$ then have the following representation as differential operators :
\ie 
J^{A B}=\eta^{B C} Z^A \frac{\partial}{\partial Z^C}-\eta^{A C} Z^B \frac{\partial}{\partial Z^C} \;,
\fe 
which can be shown to satisfy 
\ie \relax
[J^{AB},J^{CD}]=\eta^{BC}J^{AD}-\eta^{DB}J^{AC}-\eta^{AC}J^{BD}+\eta^{AD}J^{BC} \;.
\fe
The quadratic Casimir will be especially useful, since it is proportional to the Laplacian in the present representation. This takes the form 
\ie 
J^2=\eta_{A C} \eta_{B D} J^{A B} J^{C D}.
\fe

It will be convenient to use the generalized vielbein $\mathcal{E}$, as defined in \cref{sec:gennarain}, to define a basis for $\Lambda$ where $Q$ is diagonal. In this basis, the Jacobi theta function takes the form 
\ie \label{jactoav2}
\vartheta_h(\tau, z)=\sum_{\tl} \exp \left(i \pi\left\{\tau \mathbb{I}_L(\tl)-\bar{\tau} \mathbb{I}_R(\tl)\right\}\right) e^{2 \pi i [\mathbb{I}_L(\tilde{z},\tl)- \mathbb{I}_R(\tilde{z},\tl)]} \;,
\fe
where ${\tl} =\mathcal{E}(\ell +h)$, $\tilde{z} =\mathcal{E}z$, and where $\mathbb{I}_L$ and $\mathbb{I}_R$ are defined such that $Q_L=\mathcal{E}^T \mathbb{I}_L\mathcal{E}$ and $Q_R=\mathcal{E}^T \mathbb{I}_R\mathcal{E}$. From \eqref{eq:Qdecomp}, we may deduce that $\mathbb{I}_L= (\mathbbm{1} +\mathbbm{1}_{p,q})/2$ and $\mathbb{I}_R= (\mathbbm{1} -\mathbbm{1}_{p,q})/2$. Thus, we may rewrite \eqref{jactoav2} as 
\ie \label{jactoav3}
    \vartheta_h(\tau, z)
    =\sum_{\tl} \exp \left(i \pi\left\{\tau \delta_{ab}\tl^a_L\tl^b_L-\bar{\tau} \delta_{\mbar \nbar}\tl^{\mbar}_R\tl^{\nbar}_R\right\}\right) 
                      e^{2 \pi i [\delta_{ab}\tz^a_L\tl^b_L-\delta_{\mbar \nbar}\tz^{\mbar}_R\tz^{\nbar}_R ]} \;,
\fe
where we denote the $p$- and $q$-dimensional projections of $\tl$ and $\tz$ using the subscripts $L$ and $R$ respectively, and employ the indices $a, b=1,2, \ldots p, \mbar, \nbar=1,2, \ldots q$ as before.

Let us first consider the case where all the chemical potentials are set to zero. The vectors ${\tl}=(\tl_L,\tl_R)$ transform as contravariant vectors under $\rmO(p,q,\mathbb{R})$ rotations, and we can combine each of them into an $\rmO(p,q,\mathbb{R})$ vector $\tl =(\tl^A)=(\tl^{{a}}=\tl_L^{a}, \tl^{\mbar}=\tl_R^{\mbar})$.
Let us define a second-order differential operator $J^2$ by
\ie \label{cas1}
J^2:=L^a{}_b L_a{}^{b}+R^{\mbar}{}_{\nbar} R_{\mbar}{}^{\nbar}-2 T^a{}_{\nbar} T_a{}^{\nbar} \;,
\fe
with
\ie
    L^a{}_{b}:=\tl_L^a \frac{\partial}{\partial \tl_L^b}-\tl_{Lb} \frac{\partial}{\partial \tl_{La}}\;, 
    \quad  
    R^{\mbar}{}_{\nbar}:=\tl_R^{\mbar} \frac{\partial}{\partial \tl_R^{\nbar}}-\tl_{R\nbar} \frac{\partial}{\partial \tl_{R\mbar}}\;, 
    \quad 
    T^a{}{}_{\nbar}:=\tl_L^a \frac{\partial}{\partial \tl_R^{\nbar}}+\tl_{R\nbar} \frac{\partial}{\partial \tl_{La}}\;.
\fe
With the definition 
$\Delta_{\mathcal{M}_Q}:=- J^2/8$, we can show that, when all the chemical potentials, $z$, are set to zero,
\ie  \label{difff1}
\left(-\tau_2^2\left(\partial_2^2+\partial_1^2\right)-\frac{(p+q) \tau_2}{2} \partial_2-i \frac{(q-p) \tau_2}{2} \partial_1+\Delta_{\mathcal{M}_Q}\right) \vartheta_{Q, h}(\tau, \bar{\tau} ; m)=0 \;.
\fe

When we turn on chemical potentials, the definition of the $\rmO(p,q; \mathbb{R})$ generators ought to be extended:
\ie \label{hatgen}
    &\hat{L}^a{}_{b}=\tl_L^a \frac{\partial}{\partial \tl_L^b}-\tl_{Lb} \frac{\partial}{\partial \tl_{La}}+\tz_L^a \frac{\partial}{\partial \tz_L^b}-\tz_{Lb} \frac{\partial}{\partial \tz_{La}}\;,\\ 
    \quad 
    & \hat{R}^{\mbar}{}_{\nbar}=\tl_R^{\mbar} \frac{\partial}{\partial \tl_R^{\nbar}}-\tl_{R\nbar} \frac{\partial}{\partial \tl_{R\mbar}}+\tz_R^{\mbar} \frac{\partial}{\partial \tz_R^{\nbar}}-\tz_{R\nbar} \frac{\partial}{\partial \tz_{R\mbar}}\;,\\ 
   & \hat{T}^a{}_{\nbar}=\tl_L^a \frac{\partial}{\partial \tl_R^{\nbar}}+\tl_{R\nbar} \frac{\partial}{\partial \tl_{La}}+\tz_L^a \frac{\partial}{\partial \tz_R^{\nbar}}+\tz_{R\nbar} \frac{\partial}{\partial \tz_{La}}.
\fe
These generators all annihilate the $\rmO(p,q,\mathbb{R})$ invariant combination $\delta_{ab}\tz_L^a \tl_L^b-\delta_{\mbar \nbar}\tz_R^{\mbar} \tl_R^{\nbar}$. Defining the Casimir operator $\hat{J}$ for the generators given in \eqref{hatgen} by analogy with \eqref{cas1}, we now have 
\ie  \label{difff2}
\left(-\tau_2^2\left(\partial_2^2+\partial_1^2\right)-\frac{(p+q) \tau_2}{2} \partial_2-i \frac{(q-p) \tau_2}{2} \partial_1+\Delta_{\mathcal{M}_Q}\right) \vartheta_{Q, h}(\tau, \bar{\tau},\tz ; m)=0 \;,
\fe
since the statement in the previous sentence means that the proof of \eqref{difff2} follows from that of \eqref{difff1}.

Since the ensemble average is over $\rmO_Q(p,q; \mathbb{Z}) \backslash \rmO_Q(p,q;\mathbb{R})$, the result of the averaging procedure will depend only on the chemical potentials via the $\rmO(p; \mathbb{R}) \times \rmO(q; \mathbb{R})$ invariant combinations $\tz_L^2=\delta_{ab}\tz_L^a\tz_L^b$ and $\tz_R^2=\delta_{\mbar \nbar}\tz_R^{\mbar}\tz_R^{\nbar}$. Notably, the terms that depend on $z$ in $\hat{L}$ and $\hat{R}$ annihilate $\tz_L^2$ and $\tz_R^2$.  As a result, after averaging, only the potential-dependent parts of $\hat{T}^a{}_{\nbar}$ in $\Delta_{\mathcal{M}}$ contribute to the differential equation satisfied by the average. This allows us to obtain the average, denoted $f(\tau,z)$, by solving 
\ie 
\bigg(&-\tau_2^2\left(\partial_2^2+\partial_1^2\right)-\frac{(p+q) \tau_2}{2} \partial_2-i \frac{(q-p) \tau_2}{2} \partial_1
\\&
    + \frac{1}{4}\left(\tz_L^a \frac{\partial}{\partial \tz_R^{\mbar}}+\tz_{R\mbar} \frac{\partial}{\partial \tz_{La}}\right)\left(\tz_{La} \frac{\partial}{\partial \tz_{R\mbar}}+\tz_{R}^{\mbar} \frac{\partial}{\partial \tz^a_{L}}\right)
\bigg)f(\tau,\tz)=0 \;.
\fe 

To achieve this, we first note that 
\ie 
    &\left(\tz_L^a \frac{\partial}{\partial \tz_R^{\mbar}}+\tz_{R\mbar} \frac{\partial}{\partial \tz_{La}}\right)
      \left(\tz_{La} \frac{\partial}{\partial \tz_{R\mbar}}+\tz_{R}^{\mbar} \frac{\partial}{\partial \tz^a_{L}}\right)
     e^{i \pi\left(\frac{c \tz_L^2}{c \tau+d}-\frac{c \tz_R^2}{c \bar{\tau}+d}\right)}\\
&\quad=\left( \frac{4\pi c^2 \tau_2}{(c\tau +d )(c\bar{\tau} +d )}\right)^2 \tz_R^2 \tz_L^2 e^{i \pi\left(\frac{c \tz_L^2}{c \tau+d}-\frac{c \tz_R^2}{c \bar{\tau}+d}\right)}\\
&\quad  +\left( \frac{4\pi c^2 \tau_2}{(c\tau +d )(c\bar{\tau} +d )}\right)(\tz_L^2q +\tz_R^2p )e^{i \pi\left(\frac{c \tz_L^2}{c \tau+d}-\frac{c \tz_R^2}{c \bar{\tau}+d}\right)}.
\fe
In addition, we find
\ie
&\left(-\tau_2^2\left(\partial_2^2+\partial_1^2\right)-\frac{(p+q) \tau_2}{2} \partial_2-i \frac{(q-p) \tau_2}{2} \partial_1 \right) \frac{e^{i \pi\left(\frac{c \tz_L^2}{c \tau+d}-\frac{c \tz_R^2}{c \bar{\tau}+d}\right)}}{(c\tau+d)^{\frac{p}{2}}(c\bar{\tau}+d)^{\frac{p}{2}} }\\
=& -\left[ 4 \left( \frac{\pi c^2 \tau_2}{(c\tau +d )(c\bar{\tau} +d )}\right)^2 \tz_R^2 \tz_L^2 +\left( \frac{\pi c^2 \tau_2}{(c\tau +d )(c\bar{\tau} +d )}\right)(\tz_L^2q +\tz_R^2p )\right]\frac{e^{i \pi\left(\frac{c \tz_L^2}{c \tau+d}-\frac{c \tz_R^2}{c \bar{\tau}+d}\right)}}{(c\tau+d)^{\frac{p}{2}}(c\bar{\tau}+d)^{\frac{p}{2}} } \;.
\fe

This allows us to propose a solution for the ensemble average.
Recall that $\tilde{z}_R^I$ and $\tilde{z}_L^I$ are defined in the basis where $Q$ is diagonal, via the generalized vielbein $\mathcal{E}$. 
Returning to our original basis using $\delta_{ab}\tz^a_L\tz^b_L= Q_L(z)$ and $\delta_{\mbar \nbar}\tz^{\mbar}_R\tz^{\nbar}_R= Q_R(z)$,  we claim that the average can be described in terms of an Eisenstein-Jacobi series  that takes the form
\ie \label{eisjac}
E_{Q, h}(\tau,z):=\delta_{h \in \Lambda}+\sum_{(c, d)=1, c>0} \frac{\gamma_{Q, h}(c, d)\, e^{ i\pi \left(\frac{c}{c \tau +d} Q_L(z) -\frac{c}{c \bar{\tau} +d} Q_R(z)\right)}}{(c \tau+d)^{\frac{p}{2}}(c \bar{\tau}+d)^{\frac{q}{2}}} \;.
\fe

The justification for this claim is the uniqueness of this solution, which follows by comparison with the case where all chemical potentials are set to zero, as well as the fact that the Jacobi theta function (and its ensemble average) obeys a pair of heat equations. 
In the basis where $Q$ is diagonal, these equations take the form 
\ie 
\frac{\partial}{\partial \tau} f( \tz,\tau)=\frac{1}{4 \pi i} \nabla_L^2 f(\tz, \tau) \;, 
\quad 
\frac{\partial}{\partial \bar{\tau}} f( \tz, \tau)=-\frac{1}{4 \pi i} \nabla_R^2 f( \tz, \tau) \;,
\fe
where
\ie 
    \nabla_L^2=\frac{\partial^2}{\partial \tz_{Li} \partial \tz_L^i}\;,
    \quad 
    \nabla_R^2=\frac{\partial^2}{\partial \tz_{RI} \partial \tz_R^I} \;.
\fe 
Since the dependence of the average is only via $\tz_L=\sqrt{\tz_L^2}$ and $\tz_R=\sqrt{\tz_R^2}$, the Laplace operators can be written in spherical coordinates, with the angular parts acting trivially : 
\ie \label{nosphe}
    \frac{\partial f}{\partial \tau}=\frac{1}{4 \pi i}\left[\frac{\partial^2 f}{\partial \tz_L^2}+\frac{p-1}{\tz_L} \frac{\partial f}{\partial \tz_L}\right]\;,
    \quad \frac{\partial f}{\partial \bar{\tau}}=-\frac{1}{4 \pi i}\left[\frac{\partial^2 f}{\partial \tz_R^2}+\frac{q-1}{\tz_R} \frac{\partial Y}{\partial \tz_R}\right] \;.
\fe
Given that it is possible to expand the solution of these equations as 
$\sum_{m, n=0}^{\infty} f_{m, n}(\tau, \bar{\tau}) \tz_L^{2 m} \tz_R^{2 n}$, and since 
\ie 
f(\tau,\tz=0) = \delta_{h \in \Lambda}+\sum_{(c, d)=1, c>0} \frac{\gamma_{Q, h}(c, d)}{(c \tau+d)^{\frac{p}{2}}(c \bar{\tau}+d)^{\frac{q}{2}}} \;,
\fe 
we find that any possible term, denoted $\Delta f$,  that could be added to the solution \eqref{eisjac} ought to have no constant term in its expansion, that is, $\Delta f_{0,0}=0$. The equations \eqref{nosphe} imply recursion relations relating $f_{m+1,0}$ and $f_{m,0}$, as well as  $f_{0,n+1}$ and $f_{0,n}$, whereby we find that $\Delta f$ must be zero. Thus, \eqref{eisjac} is indeed the unique solution to the differential equation that remains after ensemble averaging, and is the result for the ensemble average of the Jacobi theta function \eqref{jactoav} over $\rmO_Q(p,q,\mathbb{Z}) \backslash \rmO(p,q,\mathbb{R})$.

To show that the expression \eqref{eisjac} indeed satisfies the correct modular transformation property, we note that for  
\ie 
f(M,\tau,z)=e^{ i\pi \left(\frac{c}{c \tau +d} Q_L(z) -\frac{c}{c \bar{\tau} +d} Q_R(z)\right)}
\fe
we have 
\ie 
f(M_1,\tau,z)f(M_2,\tau_{M_1},z_{M_1\tau})=f(M_2M_1, \tau,z )\;,
\fe
which is the analogue of the composition law \eqref{M_consistent}. 

Although we shall not investigate this in detail, one can in principle go on to study the ensemble average of orbifolded Narain CFTs with the inclusion of chemical potentials for global symmetries in the partition function.
\section{Comments on the Holographic Bulk}\label{sec:bulk}

We next comment on the bulk holographic dual. While we leave a detailed analysis for future work,
we will here comment on the salient features of the bulk theory.\footnote{See \cite{Benjamin:2021wzr} for discussion of orbifolded holography for 
some special examples of even self-dual $Q$.}

Let us consider the holography for the orbifolded generalized Narain theory after the ensemble average.
First, the orbifold of the Narain theory by a symmetry group $G$ will correspond to a gauging of the global $G$-symmetry of the 
bulk. We expect that the resulting theory is still a TQFT, whose anyon content will described by the 
enlarged discriminant group \cref{hat_D}, which contains the gauge-equivalence class of anyons of the original theory
as well as newly introduced anyons due to $G$-gauging. 
The modular transformations \eqref{TS_I_averaged} of the theory
determine the action of the modular group on the Hilbert space of the bulk theory on the boundary two-torus.

The expression of the Eisenstein series as a Poincar\`{e} sum \eqref{poincare_1} 
can be interpreted as a sum over $\PSL(2, \mathbb{Z})$ black holes \cite{Maldacena:1998bw, Dijkgraaf:2000fq},
which contain thermal AdS$_3$ and the BTZ black hole \cite{Banados:1992wn} as special cases. 
As in the discussion of the un-orbifolded case \cite{Ashwinkumar:2021kav},
the coefficients $\gamma^I_{(\alpha, \beta; m, n)}(g)$ in \cref{poincare_1} can be regarded as an overlap of the 
two wave functions on the solid torus, which functions are determined by the data $(\alpha, m)$ and $(\beta, n)$,
and $g$ represents the gluing of the two solid tori by an element $g$ of the mapping class group,
where $g$ here should be regarded as an element of $\PSL(2, \mathbb{Z})/\Gamma_{\infty} = \PSL(2, \mathbb{Z})/ \mathbb{Z}$. The resulting overlap is the 
partition function of the bulk TQFT
on the lens space, where the choice of the lens space is determined by $g$ and we have in general non-trivial Wilson line insertions each of the solid tori, resulting in the Hopf link invariant.
Thus the ensemble average of the boundary torus partition function represents a sum over an infinite set of 
lens space invariants of the bulk theory.
 
\section{Summary and Discussion}\label{sec:discussion}

In this paper, we discussed a class of generalized Narain CFTs associated with 
an even integral quadratic form, studied previously in \cite{Ashwinkumar:2021kav}.
We discussed orbifolds of the generalized Narain CFTs, in a procedure similar to the 
orbifolding of CFT associated with toroidal compactification of string theory.
We then discuss the ensemble averages of the orbifolded theories,
and we obtained a new generalization of Eisenstein series.

There are several obvious tasks to be pursued in the future. One problem is to identify precise holographic duals for generalized Narain theories after ensemble average, along the lines of Ref.~\cite{Afkhami-Jeddi:2020ezh,Maloney:2020nni,Ashwinkumar:2021kav}. One may also consider finding roles of such theories in string theory and more generally quantum gravity (cf.\ \cite{Ashwinkumar:2023jtz}). A direct extension of this paper is to odd quadratic forms, where we expect a spin CFT, i.e., a CFT dependent on the choice of the spin structure (See Ref.~\cite{Ashwinkumar:2021kav} discussions on the un-orbifolded case.).  Generalization of this work to asymmetric orbifolds \cite{Narain:1990mw} and CHL orbifolds would be an obvious next step. A recent topic of interest seems to be the study of relationships between code CFTs and quantum error corrections, see e.g.\ \cite{Dymarsky:2020bps,Dymarsky:2020qom,Kawabata:2023iss,Afkhami-Jeddi:2020ezh,Aharony:2023zit,Barbar:2023ncl}. It would be an interesting undertaking to identify the points of the moduli space where the generalized Narain CFT becomes rational (cf.\ \cite{Wendland:2000ye,Kidambi:2022wvh,Furuta:2023xwl}).

\section*{Acknowledgements}

We would like to thank Matthew Dodelson for collaboration in early stages of this project. We would also like to thank Florent Baume, Weiguang Cao, Craig Lawrie, Jacob McNamara, Yuto Moriwaki, Caner Nazaroglu, and Yunqin Zheng for useful discussions. 

A.K.'s research is/was supported by the Schrödinger Junior Research Fellowship (current), Riemann Fellowship (2022). A.K. thanks the hospitality of the Riemann Center for Geometry and Physics, the Simons Center for Geometry and Physics (during the Number Theory and Physics 2022 workshop), the Abdus Salam ICTP, and the Pollica Center for Physics (during the workshop on new connections between number theory and physics). M.Y.\ would like to thank KITP, Santa Barbara for hospitality during the Integrable 22 workshop. A.K. and M.Y thank the Isaac Newton Institute for Mathematical Sciences for support and hospitality during the programme ``Black Holes: Bridges betweeen number theory and holographic quantum information", and were supported by EPSRC Grant Nr. EP/R014604/1. J.M.L. thanks Kavli Institute for the Physics and Mathematics of the Universe for hospitality during the completion of a portion of this work. A.K and M.A. thank Deutsches Elektronen-Synchrotron (DESY) for hospitality during the completion of a portion of this work.
M.A.\ and M.Y.\ are supported in part by the JSPS Grant-in-Aid for Scientific Research (19H00689, 20H05860).
J.M.L. is supported by the Deutsche Forschungsgemeinschaft under Germany's Excellence Strategy - EXC 2121 ``Quantum Universe''. 
M.Y.\ is also supported in part by the JSPS Grant-in-Aid for Scientific Research (19K03820, 23H01168), and by JST, Japan (PRESTO Grant No.\ JPMJPR225A, Moonshot R\&D Grant No.\ JPMJMS2061).

\appendix

\section{Conventions and Notations}\label{appendix:conventions}

We list conventions, notations and standard definitions.
\begin{itemize}
    \item $\Lambda$: discrete data defining a generalized Narain theory,
    i.e.\ integer lattice of signature $(p,q)$ endowed with a $(p+q)$-ary quadratic form $Q$.
    \item $\Lambda^*$: dual lattice of $\Lambda$ 
    \item $\disq$: discriminant group of $Q$, $\disq= \Lambda^*/\Lambda$ and $\lvert\disq\rvert = \lvert\text{det}(Q)\rvert$
    \item $\repsp$: representation space of $Q$, such that $H\in\repsp \rightsquigarrow HQ^{-1}H = Q$
    \item $H$: Majorant associated to quadratic form $Q$, such that $HQ^{-1}H = Q$
    \item $\nviel$: generalized Narain vielbein defined by $Q= \nviel^T \mathbbm{1}_{p,q} \nviel$  and defines $H = \nviel^T\nviel$ 
    \item $\mlatt$: generalized Narain momentum lattice built from $\Lambda$ and $\nviel$
    \item $\mlatt^*$: dual lattice of $\mlatt$
    \item $\genTD$: T-duality group of entire generalized CFT, i.e.\ elements preserving $Q$
    \item $\genTDa$: T-duality group associated to a specific anyon, i.e.\ elements preserving $Q$ and the class $[\alpha] \in\disq$
    \item $\mathfrak{p}$: element of $\mlatt$
    \item $\modsp$: generalized Narain moduli space
    \item $\mlattn$: (standard) Narain lattice
    \item $\modspn$: (standard) Narain Moduli space
    \item $\narTD$: T-duality group of the (standard) Narain CFT
\end{itemize}

\section{T-dualities}
\label{appendix:TDproofs}
Here we prove a few statements in the main text. 

\begin{lemma}\label{lem:autodual}
    Every element of $\genTD$ defines an automorphism of  $\Lambda^*$.
\end{lemma}

\begin{proof}
Let $\Sigma\in\genTD$. First, we note that for every $\alpha_1\in\Lambda^*$, there exists another $\alpha_2\in \Lambda^*$ such that $\alpha_1 = \Sigma\alpha_2$. This follows from the fact that $\Sigma^{-1}$ (which is integral and unimodular) exists and 
\begin{equation}
        Q(\Sigma^{-1}\alpha,\ell) =\alpha^T \Sigma^T{}^{-1} Q \ell = \alpha^T Q \Sigma \ell \in \mathbb{Z}
\end{equation}
for all $\ell\in\Lambda$, so that $\alpha_2:=\Sigma^{-1} \alpha_1\in \Lambda^{*}$ for $\alpha_1\in \Lambda^{*}$.

Next, we note that if $\alpha_1 = \Sigma \alpha_2$, $\alpha_1^\prime = \Sigma\alpha_2^\prime$, and $\alpha_1\neq \alpha_2$, then $\alpha_1^\prime\neq \alpha_2^\prime$. This follows by contradiction: assume $\alpha_1^\prime =  \alpha_2^\prime$. Then 
\begin{equation}
    \alpha_1 - \alpha_2 = \Sigma^{-1}(\alpha_1^\prime - \alpha_2^\prime) = 0 \;.
\end{equation}
This implies $\alpha_1 = \alpha_2$, which is a contradiction. 
Thus $\Sigma\in\genTD$ defines a bijective map from $\Lambda^*$ to itself.

\end{proof}

\begin{lemma}
       Every element of $\genTD$ defines an automorphism of  $\Lambda$.
\end{lemma}

\begin{proof}
    The proof is essentially identical to that of~\cref{lem:autodual}.
\end{proof}
These give us the following corollaries immediately:
\begin{cor}
    If $\alpha\in\Lambda^*$ such that $[\alpha]=[0]$, then $[\Sigma\alpha]=[0]$ $\forall\;\Sigma\in\genTD$.
\end{cor}

\begin{cor}
   Let $\alpha_1,\alpha_2\in\Lambda^*$ such that $[\alpha_1]\neq [\alpha_2]$. Then $[\Sigma\alpha_1]\neq[\Sigma\alpha_2]$. Furthermore, $\Sigma\alpha_1\notin\Lambda$ unless $[\alpha_1]=[0]$.
\end{cor}

\begin{theorem} Then the elements of $\genTD$ define automorphisms of $\disq$ that leave $[0]$ invariant. 
\end{theorem}

\begin{proof}
The proof follows from the above Lemmas and Corollaries.
\end{proof}

\section{Orbifold Jacobi Theta Functions} \label{appendix:theta}

In this section, we derive the transformation properties of 
the orbifold theta function \eqref{def_Theta_I} and the Jacobi theta function \eqref{jactoav}. While we include these for completeness of the presentation,
we note that essentially the same transformations were derived previously in a 
recent mathematics paper \cite{MR4309840} (see also \cite{MR1625724}).

Since $S$ and $T$ are generators of $\SL(2, \mathbb{Z})$, one can verify the general transformation \eqref{theta_SL2}
by checking that (1) the formula coincides with \eqref{Theta_T}, \eqref{Theta_S} for $T, S$-transformations and that (2) composes consistently under the composition of two $\SL(2, \mathbb{Z})$ matrices. The latter statement follows from \eqref{M_consistent_I}, which is known for the 
un-orbifolded case since the times of C.~Siegel. In the following we will nevertheless directly prove the 
modular transformation under a general element of $\SL(2, \mathbb{Z})$.

\subsection{Modular Transformations}

We shall now consider general modular transformations for Jacobi theta functions with characteristics, associated with an arbitrary even integral lattice, denoted $I$, with indefinite signature $(p_I,q_I)$: 
\ie 
 \vartheta_{\delta, \eta}(\tau,z)=\sum_{y=\ell+\delta} \exp \left(i \pi\left\{\tau Q_L(y)-\bar{\tau} Q_R(y)\right\}\right) e^{2 \pi i Q(y, \eta)}e^{2 \pi i Q(z,y)}.
 \fe 

Under a modular transformation by $[M]=\left[\left(\begin{array}{cc} a& b \\ c&d \end{array}\right) \right] \in \PSL(2, \mathbb{Z})$:
\ie
    \tau &\rightarrow \tau_{M} := \frac{a\tau+ b }{ c  \tau+ d  } = \frac{a}{ c }+ c ^{-2} \frac{-c}{c\tau +d} 
    \;,\\
    z &\rightarrow z_{M,\tau}:=Q^{-1}\left( \frac{1}{c \tau +d }Q_L-\frac{1}{c \bar{\tau} +d}Q_R\right)z,
\fe
where we assumed $c\ne 0$.
To simplify the notation, in this appendix we denote $\vartheta_{I, (\delta, \eta)}$ by $\vartheta_{\delta, \eta}$.

We begin with the expression
\begin{align}\label{theta_begin}
     \vartheta^I_{\delta,\eta}(\tau_{M},z_{M,\tau}) &= \sum_{\ell \in I} e^{i\pi \left[ \frac{a}{ c }Q(\ell+\delta) +  c ^{-2}U(\ell+\delta) \right]} e^{2\pi i Q(\ell+\delta,\eta)}e^{-\frac{2\pi i }{c} U(z,\ell +\delta)} \;,
\end{align}
where we defined 
\begin{align}
    \tau_M Q_L - \bar{\tau}_MQ_R &= \frac{a}{ c }Q + c ^{-2}U \;, \quad
    U :=  \frac{-c Q_L}{c\tau+d} + \frac{c Q_R}{c\taubar+d} \;.
    \label{U_def}
 \end{align}
Note that $U$ depends on the moduli $\tau, \taubar$.
We now write $\ell =  c  \ell'+\ell_c $ in the sum \eqref{theta_begin}, where $\ell' \in I$ and $\ell_c \in I/( c  I)$. After massaging the resulting expression, 
we obtain\footnote{Here $\cdot$ denotes a multiplication for matrices and vectors. For example, in the expression $U^{-1}[Q\cdot ( c \eta +a\delta)]$, the combination $Q\cdot ( c \eta +a\delta)$ denotes the 
product of a matrix $Q$ and a vector $c \eta +a\delta$, and the resulting vector is then fed into the argument of the quadratic form $U^{-1}$.} 
\begin{align}
     \vartheta^I_{\delta,\eta}(\tau_M,z_{M,\tau})
            &=   \sum_{\ell_c  \in I/( c  I)} e^{\left(i\pi\frac{a}{ c }Q( \ell_c +\delta)\right)} e^{\left(-\frac{2\pi ia}{ c }Q(\ell_c ,\delta)-\frac{2\pi ia}{ c }Q(\delta)\right)}\nonumber\\
                &\hspace{2cm} \times\sum_{\ell'\in I} e^{\left( i\pi  U\left[ \ell' +  c ^{-1}(\ell_c +\delta)+(U^{-1}Q)\cdot ( c \eta +a\delta)-z\right] \right)}e^{-i\pi U^{-1}[Q\cdot ( c \eta +a\delta)]}\nonumber \\
                & \hspace{2cm} \times e^{2\pi i Q(c \eta +a\delta, z)} e^{-i\pi U(z,z) }.
\end{align}
Performing a Poisson resummation in $\ell'\in I$, we obtain\footnote{After a Poisson resummation, we obtain a sum over an element $\ell^*$ of the dual lattice $I^*$,
which we can write as $\ell^*= \ell + \gamma$ with $\ell \in I, \gamma\in I^*/I$.} 
\begin{align} 
     \vartheta^I_{\delta,\eta}(\tau_M,z_{M,\tau})&= [\textrm{det}(-iU)]^{-\frac{1}{2}}\sum_{\ell_c  \in I/( c  I)} e^{\left(i\pi\frac{a}{ c }Q( \ell_c +\delta)\right)} e^{\left(-\frac{2\pi ia}{ c }Q(\ell_c ,\delta)-\frac{2\pi ia}{ c }Q(\delta)\right)}  e^{-i\pi U^{-1}[Q( c \eta +a\delta)]}\nonumber\\
                & \times \sum_{\ell \in I} \sum_{\gamma\in \mathscr{D}_I} e^{\left( -i\pi  U^{-1}[Q\cdot (\ell + \gamma)]-2\pi i \left(Q\cdot (\ell +\gamma )\right)^T \cdot \left\{ c ^{-1}(\ell_c+\delta)+(U^{-1}Q)\cdot ( c \eta +a\delta)-z\right\}  \right) }
                \;. \nonumber \\
                & \times e^{2\pi i Q(c \eta +a\delta, z)} e^{-i\pi U(z,z) }
\end{align}
Since we have (as follows from the definition of $U$ in \eqref{U_def})
\begin{align}
    -U^{-1} [Q\cdot y]&= \bigg(\frac{ d  }{ c }Q + (Q_L\tau - Q_R\bar{\tau})\bigg) [y] \;,
\end{align}
and 
\ie 
e^{ -i\pi U(z,z)}=  e^{ i\pi \left(\frac{c}{c \tau +d} Q_L(z) -\frac{c}{c \bar{\tau} +d} Q_R(z)\right)},
\fe 
we can rewrite
\begin{align} \label{large_sum}
     &\vartheta^I_{\delta,\eta}(\tau_M,z_{M,\tau })= [\textrm{det}(-iU)]^{-\frac{1}{2}}\left(\sum_{\ell_c  \in I/( c  I)} e^{\left(i\pi\frac{a}{ c }Q( \ell_c +\delta)\right)} e^{\left(-\frac{2\pi ia}{ c }Q(\ell_c ,\delta)-\frac{2\pi ia}{ c }Q(\delta)\right)} \right) \nonumber\\
                & \times e^{i\pi \frac{d}{c} Q( c \eta +a\delta)}  e^{i\pi (\tau Q_L - \taubar Q_R)( c \eta +a\delta)}  e^{ i\pi \left(\frac{c}{c \tau +d} Q_L(z) -\frac{c}{c \bar{\tau} +d} Q_R(z)\right)} \nonumber\\
                & \times \sum_{\ell \in I} \sum_{\gamma\in \mathscr{D}_I} 
               e^{\left[ i\pi  \frac{d}{c}(\ell + \gamma) + i\pi  (\tau Q_L - \taubar Q_R) (\ell+\gamma) 
                -\frac{2\pi i}{c} Q(\ell +\gamma ,\ell_c+\delta) 
                 +\frac{2\pi i d}{c} Q (\ell +\gamma, c \eta +a\delta)
                 + 2\pi i (\tau Q_L - \taubar Q_R)  (\ell +\gamma, c \eta +a\delta)
                \right] }
                \; \nonumber
                \\ & \times e^{2i\pi Q(\ell + \gamma +c \eta +a\delta, z)} .
\end{align}
While the expression \eqref{large_sum} looks formidable, the 
moduli-dependent terms
combine nicely into
\begin{align}
e^{ i\pi (Q_L\tau - Q_R\bar{\tau})[\ell+\gamma+ c \eta+a\delta ] }  e^{ i\pi \left(\frac{c}{c \tau +d} Q_L(z) -\frac{c}{c \bar{\tau} +d} Q_R(z)\right)}\;,
\end{align}
while the moduli-independent exponential terms combine into
\begin{align}\label{combination}
e^{2i\pi Q(\ell + \gamma +c \eta +a\delta, z)} e^{2i\pi Q[\ell +\gamma+ c  \eta  + a \delta,  d   \eta +  b \delta] } \,
 e^{-i\pi Q(a\delta+ c \eta, b\delta + d \eta) }
 \left(\sum_{\ell_c \in I/( c  I)} e^{\frac{i\pi}{ c }\left(a Q[\ell_c ] - 2Q[\ell_c , \gamma] 
                    +  d   Q[\gamma]\right) }  \right) \;.
\end{align}
This combination is preserved when we shift $\gamma$ 
by an element of $I$, so that we can set $\ell=0$ in \eqref{combination}. 
This means we can bring \eqref{combination} outside the sum over $\ell$,
and we find that the sum over $\ell$ over the remaining $\ell$-dependent expression gives rise to the Jacobi theta function  $\vartheta^I_{(\gamma +a\delta+ c \eta, b\delta + d  \eta)}(\tau,z)$.
Finally, we need to include the remaining factor
\begin{align}
[\textrm{det}(-iU)]^{-\frac{1}{2}} = \left[\textrm{det}\left(\frac{-ic Q_L}{c\tau+d} \right) \right]^{-\frac{1}{2}} \left[ \textrm{det}\left(\frac{ic Q_R}{c\taubar+d}\right)  \right]^{-\frac{1}{2}} 
= \frac{ e^{-\frac{i\pi}{4}(p_I-q_I)} c^{-\frac{p_I+q_I}{2}} \left(\textrm{det}\, Q\right)^{-\frac{1}{2}}}{(c\tau+d)^{-\frac{p_I}{2} }   (c\taubar+d)^{-\frac{q_I}{2} } } \;.
\end{align}
  Shifting $\delta$ by an element $\gamma'\in \mathscr{D}_I$, and defining $\tilde{\gamma}=\gamma + a \gamma' \in \Lambda^*/\Lambda$ allows us to retrieve the most general modular transformation. By combining all the factors above, we obtain
\ie 
\vartheta_{(\delta+\gamma', \eta)}\left(\tau_M,z_{M,\tau}\right)&=\frac{ e^{-\frac{i\pi}{4}(p_I-q_I)} c^{-\frac{p_I+q_I}{2}} \left(\textrm{det}\, Q\right)^{-\frac{1}{2}}}{(c\tau+d)^{-\frac{p_I}{2} }   (c\taubar+d)^{-\frac{q_I}{2} } }\left(e^{-i \pi ab Q[\delta]-2 i \pi bc Q[\delta, \eta]-i \pi cd Q[\eta]}\right) \\
& \quad \times 
e^{ i\pi \left(\frac{c}{c \tau +d} Q_L(z) -\frac{c}{c \bar{\tau} +d} Q_R(z)\right)} \sum_{\tilde{\gamma}} \lambda_{\gamma', \tilde{\gamma}} \vartheta_{\left(\tilde{\gamma}+c \eta+a \delta, d \eta+b\delta\right)} (\tau,z) ,
\fe
where 
\ie 
    \lambda_{\gamma',\tilde{\gamma}}
    :=\sum_{\ell_c \in I/cI}e^{\frac{i \pi}{c}\left(a Q[\ell_c+\gamma']-2 Q\left[\tilde{\gamma}, \ell_c+\gamma'\right]+dQ\left[\tilde{\gamma}\right]\right)} \;.
\fe
In particular, setting $\delta=0=\eta$, we find the transformation formula for a Jacobi theta function associated with a general even integral lattice 
\ie 
&\vartheta_{\gamma'}\left(\tau_M,z_{M,\tau}\right) =\frac{ e^{-\frac{i\pi}{4}(p_I-q_I)} c^{-\frac{p_I+q_I}{2}} \left(\textrm{det}\, Q\right)^{-\frac{1}{2}}}{(c\tau+d)^{-\frac{p_I}{2} }   (c\taubar+d)^{-\frac{q_I}{2} } }  \sum_{\gamma} \lambda_{\gamma',\gamma} \vartheta_{\gamma } (\tau,z) 
e^{ i\pi \left(\frac{c}{c \tau +d} Q_L(z) -\frac{c}{c \bar{\tau} +d} Q_R(z)\right)}.
\fe
If we instead set $z=0$, we obtain \eqref{theta_SL2}. 

We can see from \eqref{theta_SL2} that the only effects of $\delta$ and $\eta$ are to 
change the overall phases of the modular transformation. Note that in our derivation we have not imposed any conditions on $\delta$ and $\eta$.
If we consider the S-transformations with $a= d  =0$ and $ c  = - b  = 1$, we obtain
\begin{align}
	\vartheta^I_{\delta,\eta} \left(-\frac{1}{\tau}\right) &= \frac{e^{-i\pi \frac{(p_I-q_I)}{4}}}{\sqrt{\lvert \mathscr{D}_I \rvert}}\tau^{p_I/2}\bar{\tau}^{q_I/2}e^{2\pi i Q[\eta,\delta]}\sum_{\gamma\in \mathscr{D}_I}  \vartheta^I_{( \gamma +\eta,-\delta)}(\tau) \;.
\end{align}

While we assumed $c\ne 0$ in the derivation, the case $c=0$ amounts to the $T$-transformation of the form $\tau\rightarrow \tau+ab$ with $a^2=1$, and $z \rightarrow az$. In this case, it is straightforward to show that 
\ie 
    \vartheta_{(\delta+\gamma', \eta)}\left(\tau_M,z_{M,\tau}\right)
    &= e^{\pi i ab Q(\gamma')} e^{-\pi i ab Q(\delta)}\vartheta_{(\delta+\gamma', \eta +ab\delta  )}\left(\tau,az\right)\\ 
    &= e^{\pi i ab Q(\gamma')} e^{-\pi i ab Q(\delta)}\vartheta_{(a\delta+a\gamma', a\eta +b\delta  )}\left(\tau,z\right).
\fe 
In fact, the formula \eqref{theta_SL2} works for the T-transformations with $a= b = d  =1$ and $c =0$:
\begin{align}
	\vartheta^I_{\delta,\eta} (\tau+1) 
    &= e^{-i\pi Q[\delta]}\, \vartheta^I_{(\delta,\eta+\delta)}(\tau) \;,
\end{align}
which formula can be verified directly from the definition of the theta function.

\subsection{Asymptotics at Cusps}\label{cuspasymp}

At the cusp $\tau\to i\infty$, we find from the definition \eqref{def_Theta_I} that
\begin{align} \label{infinity}
\vartheta_{I, (\delta, \eta)} (\tau\to i\infty)= \delta_{\delta \in I} \;.
\end{align}
We use the transformation law of the theta functions to determine the behavior at another cusp $\tau = -d/c$.
For this purpose we use \eqref{theta_SL2} for $M^{-1} = \begin{pmatrix}d & -b\\ -c & a\end{pmatrix}$ 
\begin{align}
     \vartheta^I_{(\delta,\eta)}\left(\frac{d\tau-b}{-c\tau+a}\right)
     & = \mu_{(\delta, \eta)\cdot M^{-1}} \sum_{\gamma \in \mathscr{D}_I}\mathcal{U}_{\gamma}(M^{-1}) \,\vartheta_{(\gamma + d\delta- c \eta, -b \delta+a  \eta )}(\tau) \;,
 \end{align}
Let us define $\tau':=(d\tau-b)/(-c\tau+a)$, 
or equivalently $\tau=(a\tau'+b)/(c\tau'+d)$,
and consider taking the limit $\tau\to i\infty$,
which is equivalent to $\tau'\to -d/c$.
In this limit, using \eqref{infinity}, we obtain\footnote{
The expression on the right-hand side depends only on the element $\PSL(2,\mathbb{Z})/\Gamma_{\infty}$,
namely the expression depends only on the entries $c,d$ of the $\PSL(2, \mathbb{Z})$ matrix.
To see this, let us pick up a modular inverse $\dtil$ of the integer $d$ via $d\dtil = 1 \text{ mod } c$, so that
\begin{align}
\begin{split}
    a = ck +\dtil \;,\quad
    b = dk + \frac{d\dtil-1}{c} \;, \quad (k \in \mathbb{Z}) \;,
\end{split}
\end{align}
where $k$ labels an element of $\Gamma_{\infty}\simeq \mathbb{Z}$.
The $k$-dependence generates an extra factor of
\begin{align}
e^{i \pi k(Q(d\delta-c \eta) -Q(\beta))}\;,
\end{align}
which is equal to the identity under the delta-function constraint $\beta +d\delta -c\eta \in I$ (recall $\beta\in \mathscr{D}_I$).
This ensures that the behavior at the cusp is independent of the upper two entries of the $\PSL(2,\mathbb{Z})$ matrix.}
\begin{align} \label{cusp}
    \vartheta_{(\delta, \eta)}\left(\tau' \to -\frac{d}{c}\right) &\to \frac{1}{\sqrt{\textrm{det}\, Q_I} } \, e^{-\frac{\pi i}{4}(p_I-q_I)} 
     (c\tau'+d)^{-\frac{p}{2}}(c\taubar'+d)^{-\frac{q}{2}}  \mu_{(\delta, \eta) \cdot M^{-1}}   \lambda_{0, -d \delta + c\eta}(M^{-1}) \;,
\end{align}
where we used $-c\tau+a = 1/(c\tau'+d)$. This matches the behavior of the Eisenstein series \eqref{E_sum_1} at the cusp.

As is clear from this discussion above, the only properties needed for the derivation of the 
asymptotic behavior at the cups are (i) behavior at $\tau\to i\infty$ as in \eqref{infinity}
as well as (ii) the modular transformation rule \eqref{def_Theta_I}.
Since we can directly verify these two properties
for the Eisenstein series \eqref{E_sum_1}, we can repeat the same argument to 
conclude that the Einstein series has the same asymptotic as in \eqref{cusp}.

%
\section{Congruence Subgroup}\label{appendix:congruence}
In this appendix, we prove that the theta functions we have been working with are modular forms for the congruence subgroup $\Gamma(N^2 L)$, where $N$ is the order of the orbifold action and $L$ is an integer multiple of the level $L_Q$ of the quadratic form.
Our discussion here seems to be new for $N>1$.
Our proof relies heavily on the result of \cite{Shintani1975OnCO,Oda1977}, which discuss the un-orbifolded case $N=1$.\footnote{For the special case $N=1$, 
 \cite{Shintani1975OnCO} already noted
 that the modular transformation formulas simplify for a special congruence subgroup,
 and for example already considered $\Gamma(\textrm{lcm}(L_Q, 4))$ for $p+q$ odd.
 The focus there, however, was not necessarily to 
 identify the precise congruence subgroup where the multiplier system evaluates to $1$, and we have not found an explicit statement in the literature that our theta function is modular (without any multiplier system) for a concrete choice of $\Gamma(L)$. }

We shall focus on the case modular transformations where $c\neq 0$ first. Let us reproduce the corresponding transformation law of the orbifold theta functions: 
\begin{align}
    \vartheta_{(\alpha,0)+(\delta,\eta)}(\tau_M) 
    &= \sum_{\beta\in\mathcal{D}}   \mathcal{U}_{\alpha,\beta}(M,\tau;(\delta,\eta))\,
                                                     \vartheta_{\underbrace{(\beta,0)+(\delta,\eta)\cdot M}_{\textrm{characteristics}}} \;,\\
   \mathcal{U}_{\alpha,\beta}(M,\tau;(\delta,\eta)) &= (c\tau+d)^{\frac{p}{2}}(c\taubar+d)^{\frac{q}{2}} \nonumber\\
   &  \times \underbrace{\exp\bigg[-i\pi(ab Q[\delta]+2bc Q[\delta,\eta] +cd Q[\eta])\bigg]}_{\textrm{shift-dependent phase}}
  \underbrace{\lambda_{\alpha,\beta}(M)}_{\textrm{transformation matrix}} \;.
\end{align}
We wish to show that this reduces to the standard transformation for modular functions
when $M\in \Gamma(N^2 L)$.
We break the proof into three parts: the shift-dependent phase, the theta function characteristics, and the transformation matrix.

\paragraph{Shift-dependent Phase}
We investigate the factor
\begin{align}
   \exp\bigg[-i\pi(ab Q[\delta]+2bc Q[\delta,\eta] +cd Q[\eta])\bigg] \;.
\end{align}
For the factor of $\textrm{exp}(-i\pi (ab Q[\delta]))$, we observe that it is trivial since $b\equiv 0\textrm{ mod }N^2L$, $bQ[\delta]\in 2\mathbb{Z}$. This is because $\delta \in I/N$ and from the definition of the level,  $LQ[\ell']\in 2 \mathbb{Z}$ for $\ell'\in I^*$ (recall that $I\subset I^*$). Similarly, since $\eta \in I/N$ and $c\equiv 0\textrm{ mod }N^2L$ the factor of $\textrm{exp}(-i\pi (cd Q[\eta]))$ is trivial as well. Finally, the factor of $\textrm{exp}(-i\pi (2bc Q[\delta,\eta]))$ is trivial since $c\equiv 0\textrm{ mod }N^2L$, so that $c Q[\delta,\eta]\in \mathbb{Z}$.

\paragraph{Characteristics}
We look at the behaviour of the characteristics $(\delta, \eta)$ under a modular transformation:
\begin{align}
   (\delta,\eta)\rightarrow (\delta,\eta)\begin{pmatrix}   a & b\\
   c & d
   \end{pmatrix} = (a\delta+c\eta, b\delta +d\eta ) \;.
\end{align}
In the definition of our theta functions, we have a lattice sum over $I+\alpha+\delta$, $\alpha\in\mathcal{D}$ and $\delta\in I/N$. After the transformation, we have the sum
\begin{equation}
   \sum_{\ell\in I+\alpha+a \delta +c\eta} \exp (i \pi \tau Q_L(\ell)-i \pi \bar{\tau} Q_R(\ell)) \exp (2i\pi Q_I(\ell, \delta b+\eta d)) \;.
   \end{equation}
   Now, since 
           $a\delta \equiv \delta\textrm{ mod }I, 
           d\eta \equiv \eta \textrm{ mod }I, 
           b\delta \in I, 
           c\eta \in I$, the sum reduces to 
\begin{equation}
   \sum_{\ell\in I+\alpha+ \delta } \exp (i \pi \tau Q_L(\ell)-i \pi \bar{\tau} Q_R(\ell)) \exp (2i\pi Q_I(\ell, \eta )) \;,
\end{equation}
   and therefore
   the characteristics of the transformed theta function reduce to those of the theta function prior to modular transformation.

\paragraph{Transformation Matrix}
We now show that the matrix $\lambda^{\Lambda}_{\alpha, \beta}(M)$ 
becomes trivial (i.e.\ unity) for elements of the congruence subgroup $\Gamma(N^2L)$. 

Since we are interested in the evaluation of $\lambda_{\alpha, \beta}^{\Lambda}(M)$,
and since this factor was defined by the modular transformation property of the 
un-orbifolded theta function, we can take advantage of the literature for 
$N=1$.

In the following we use the quadratic residue symbol $\displaystyle\left(\frac{a}{b}\right)$ as defined in \cite{Shimura},
which is known as the Kronecker symbol.
This symbol is defined for an integer $a$ and an odd integer $b\ne 0$, and coincides with 
the ordinary quadratic residue symbol (Legendre symbol) when $b$ is an odd prime. This symbol has many interesting properties, 
e.g.\ $\jac{a}{b}$ for a fixed $b$ is a character modulo $b$
 as a function of $a$. See \cite{Shimura} for further properties of this symbol.

For $p+q$ even, it is known \cite{Oda1977}
that $\lambda_{\alpha, \beta}(M)=\delta_{\alpha, \beta}$
for $M\in \Gamma(L_Q)$.
The case of $p+q$ odd is more subtle:
it was shown in \cite{Oda1977} that
\begin{align}
    \lambda_{\alpha, \beta}(M)
    =\left(\frac{-2}{d} \right)
    \left(\frac{c}{d} \right)
    \left(\frac{\det\, Q}{d} \right)
    \delta_{\alpha, \beta}
\end{align}
for $M\in \Gamma(\textrm{lcm}(L_Q, 4))$.
If we further impose $c\equiv 0 \bmod \det\, Q$, we have
$\jac{c}{d}=\jac{\det\, Q}{d}$, which cancels the two factors, leaving 
\begin{align}
    \left(\frac{-2}{d} \right) = \jac{-1}{d}\cdot \jac{2}{d}~,
\end{align}
which can be evaluated using the following identities
\begin{align}
\label{eq:jacobisymbolidentities1}
\begin{split}
 \left( \dfrac{2}{d}\right) &= (-1)^{\frac{d^2-1}{8}} = \begin{cases}
         1, & d \equiv 1,7 \mod 8 \\
         -1, & d \equiv 3,5 \mod 8
     \end{cases}  \\
     \jac{-1}{d} &= (-1)^{\frac{d-1}{2}} 
      =\begin{cases}
         1, & d \equiv 1 \mod 4 \\
         -1, & d \equiv 3 \mod 4
     \end{cases} .
\end{split}
     \end{align}
These symbols trivialize for $d\equiv 1\bmod 8$,
and hence $\lambda_{\alpha, \beta}(M) = \delta_{\alpha, \beta}$.
The restrictions imposed on $c, d$ can be realized by 
choosing $M$ from a congruence subgroup $\Gamma( \textrm{lcm}(L_Q, 8, |\det\, Q|)
=\Gamma( \textrm{lcm}(8, |\det\, Q|))$.

For future reference, let us state this as a Proposition:\footnote{The implication of this result is that if we consider a Narain moduli space of conformal field theories of with difference in cleft and right moving central charges $c_L-c_R = p-q \in 2\mathbb Z$ admitting an action of an order $N$ orbifold, the theta function transforms as a modular form on $\Gamma(N^2 L)$ (which we momentarily refer to as the `even' congruence subgroup), where $L$ is the level of the Narain lattice. However, if the chiral central charge is \textit{odd}, then the theta function transforms on a smaller congruence subgroup contained in the even congruence subgroup.}
\begin{props}
\label{prod:csg}
   $\lambda_{\alpha, \beta}(M) = \delta_{\alpha, \beta}$
   if $M \in \Gamma(L_Q)$ for $p+q$ even and 
   if $M \in \Gamma( \textrm{lcm}(8, |\det\, Q|))$ for $p+q$ odd.   
\end{props}
\begin{rmk}
\Cref{prod:csg} demonstrates that the Siegel-Weil formula holds for a sufficiently small principal congruence group $\Gamma \subset \Gamma(N^2 L)$. It is however possible that the congruence subgroup can be enlarged slightly such that the Siegel-Weil theorem still holds, thereby permitting ensemble averages, but we are unaware of justifiable constraints that allow us to do this.
\end{rmk}

\subsubsection*{The $c=0$  case}

   We have so far restricted ourselves to the case of $c\neq 0$. We also need to show that the theta function is a modular function for $\Gamma(N^2L)$ even when $c=0$.    In this case, the modular transformation law takes the form 
   \ie \label{d24}
\vartheta_{(\delta+\alpha, \eta)}\left(\tau_M\right)= e^{\pi i ab Q(\alpha)} e^{-\pi i ab Q(\delta)}\vartheta_{(\delta+\alpha, \eta +ab\delta  )}\left(\tau\right),
\fe 
where $\alpha \in I^*/I$, and $\delta, \eta \in I/N$. 
Since $c=0$, and $a, d=1 \textrm{ mod } N^2L$ for an element of $\Gamma(N^2L)$, it must be the case that $a=1$ and $d=1$ since $N>1$. In particular, $c\tau +d =1$, and the theta function should be invariant under the modular transformation.

Now, since $b = 0\textrm{ mod }N^2L$, and $LQ[\ell'] \in 2 \mathbb{Z}$ for $\ell' \in I^*$, we find that the phases in \eqref{d24} are trivial. Moreover, 
   \ie 
\vartheta_{(\delta+\alpha, \eta +ab\delta  )}\left(\tau\right)=\sum_{y=\ell + \alpha + \delta} e^{\pi i (\tau Q_L[y] -\bar{\tau} Q_R[y])}e^{2\pi i Q(y,\eta)}e^{2\pi i Q(y,ab \delta)},
   \fe
   and we find that 
   \ie 
e^{2\pi i Q(y,ab \delta)}=e^{2\pi i Q(\ell + \alpha +\delta ,ab \delta)}=1
   \fe
   since $a=1$ and $b = 0\textrm{ mod }N^2L$. Thus, we find that 
   \ie 
\vartheta_{(\delta+\alpha, \eta)}\left(\tau_M\right)= \vartheta_{(\delta+\alpha, \eta  )}\left(\tau\right),
   \fe
where $M\in \Gamma(N^2 L)$ with $c=0$.
This completes a proof that the theta function is a modular function for $\Gamma(N^2 L)$.

\section{Transformation of Eta and Theta Functions}\label{appendix:modular_transformation}

We reproduce here modular transformation properties of some modular forms used in the text.

For $M=\begin{pmatrix}a & b\\ c & d \end{pmatrix}\in\SL(2,\mathbb{Z})$, the Dedekind eta function transforms as 
\begin{equation}
    \eta(M\cdot \tau) = \epsilon(M) (c\tau+d)^{\frac{1}{2}}\eta(\tau) \;,
\end{equation}
where the multiplier system $\epsilon(M)$ is given by 
\begin{align}
    \label{eq:etamultiplier}
    \epsilon(M) := \begin{cases}
        \exp\left({\dfrac{i \pi b}{12}}\right), & c = 0, \ d = 1 \\
        \exp\left(\dfrac{a+d}{24c}-\dfrac{s(d,c)}{2}-\dfrac 18 \right), & c>0
    \end{cases}~,
\end{align}
where the $s(d,c)$ is the Dedekind sum given by
\begin{align}
    s(d,c) = \sum_{n \textrm{ mod } c} ((n/c)) ((dn/c)) \;,
\end{align}
and $((x))$ is the sawtooth function 
\begin{align}
    ((x)) & := 
    \begin{cases}%
     x - [x] -\frac{1}{2} & x \in \mathbb{R}\backslash \mathbb{Z} \;,\\
     0 & x \in \mathbb{Z} \;,
    \end{cases} 
\end{align}
with $[x]$ being the integer part of $x$.\footnote{A study of Dedekind sums and their relation to the study of the eta function can be found in \cite{rademacher1972dedekind}.}

Consider the Jacobi theta function with characteristics
\begin{align}
    \vartheta\begin{bmatrix}\alpha\\  \beta \end{bmatrix}(\tau) := \sum_{n\in\mathbb{Z}} e^{i\pi (n+\alpha)^2\tau}e^{2\pi i (n+\alpha)\beta}.
\end{align}
Under an element of $M= \begin{pmatrix}a & b\\ c & d \end{pmatrix}\in\SL(2, \mathbb{Z})$ with $c>0$, we have the transformation law 
\begin{align}
    \vartheta\begin{bmatrix}\alpha + \frac{1}{2} \\ \beta + \frac{1}{2} \end{bmatrix}(\tau) 
    = \frac{\epsilon(M)^{-3} \myS(M,\alpha,\beta)^{-1}}{ \sqrt{c\tau+d}}\,  \vartheta\begin{bmatrix} a\alpha + c\beta + \frac{1}{2}\\   b\alpha + d\beta +\frac{1}{2}\end{bmatrix}(M\cdot \tau) \;,
\end{align}
where the characteristic-dependent phase is defined as in \eqref{myS_def}.
This implies
\begin{equation}
    \frac{\vartheta\begin{bmatrix}\alpha+\frac{1}{2}\\\beta+\frac{1}{2}\end{bmatrix}(\tau)}{\eta(\tau)} 
    = \epsilon(M)^{-2} \myS(M,\alpha,\beta)^{-1} \frac{\vartheta\begin{bmatrix}a\alpha + c\beta+\frac{1}{2}\\ b\alpha + d\beta+\frac{1}{2}\end{bmatrix}(M\cdot\tau)}{\eta(M\cdot\tau)} \;,
\end{equation}

\bibliographystyle{ytphys}
\bibliography{narain_orb}


\end{document}